\documentclass{article}     

\usepackage{amsmath, amsfonts,amsthm,amssymb}
\usepackage{dsfont}

\usepackage[english]{babel}
\usepackage{latexsym,amssymb,amsfonts,graphicx}
\usepackage{amsmath,dsfont}
\usepackage{verbatim}
\usepackage{mathrsfs}
\usepackage{bm}
\usepackage{color}

\usepackage{url}

\usepackage{bm}
\usepackage{natbib}

\newcommand{\Px}{ \mathbb{P} }
\newcommand{\Qx}{ \mathbb{Q} }
\newcommand{\Ex}{ \mathbb{E} }
\newcommand{\gt}{\mathcal{G}_t}

\newcommand{\idc}{\mathbf{1}}

\newcommand{\Gx}{\mathbb{G}}
\newcommand{\Fx}{\mathbb{F} }

\newcommand{\F}{\mathcal{F}}
\newcommand{\G}{\mathcal{G}}

\newcommand{\aPlus}{(a_{1,1} + a_{2,2})}

\newcommand{\CPlus}{c_{+}}
\newcommand{\CPlusSQ}{\sqrt{\CPlus}}
\newcommand{\cone}{c_1}
\newcommand{\cthree}{c_2}

\newcommand{\conethree}{(\cone+\cthree)}
\newcommand{\conethreenop}{\cone+\cthree}
\newcommand{\mcthreeone}{(\cthree - \cone)}

\newcommand{\CPlusExpr}{ \mcthreeone^2 + 4 a_{1,1} a_{2,2}}

\newcommand{\CPlusT}{ \frac{t \CPlusSQ }{2}}

\newcommand{\IntTR}{\int_t^R}

\newcommand{\CPlusSQcot}{(\CPlusSQ + \conethreenop)}

\newcommand{\termbes}{a_{12}^{\Qx} a_{21}^{\Qx} x (R-x)}

\newcommand{\CTcosh}{\cosh\left(\CPlusT\right)}
\newcommand{\CTsinh}{\sinh\left(\CPlusT\right)}

\newcommand{\invCPlus}{\frac{1}{2 \CPlus}}

\newcommand{\commterm}{  \invCPlus   e^{\frac{t   \conethree }{2}}  }
\newcommand{\ExpCPlusSQ}{ e^{s \CPlusSQ}}
\newcommand{\Expcoct}{( \ExpCPlusSQ-1)}

\newtheorem{theorem}{Theorem}[section]
\newtheorem{definition}{Definition}[section]

\newtheorem{proposition}[theorem]{Proposition}
\newtheorem{remark}[theorem]{Remark}
\newtheorem{lemma}[theorem]{Lemma}
\newtheorem{corollary}[theorem]{Corollary}

 \definecolor{Red}{rgb}{1.00, 0.00, 0.00}
    
    \definecolor{DRed}{rgb}{0.5, 0.00, 0.00}
    
    \definecolor{Blue}{rgb}{0.00, 0.00, 1.00}
    
        \definecolor{Green}{rgb}{0.0, 0.4, 0.0}

    \definecolor{PaleGrey}{rgb}{.6, .6, .6}

\title{Dynamic Portfolio Optimization with a Defaultable Security and Regime Switching}
\author{Agostino Capponi\thanks{{School of Industrial Engineering, Purdue University, West  Lafayette, IN, 47907, USA ({\tt capponi@purdue.edu}).
}}
 \and Jos\'e E. Figueroa-L\'{o}pez\thanks{Department of Statistics, Purdue University, West Lafayette, IN, 47907,  USA ({\tt figueroa@purdue.edu}).}}
\date{}

\begin{document}
\maketitle
\begin{abstract}
We consider a portfolio optimization problem in a defaultable market with finitely-many economical regimes, where the investor can dynamically allocate her wealth among a defaultable bond, a stock, and a money market account. The market coefficients are assumed to depend on the market regime in place, which is modeled by a finite state continuous time Markov process. We rigorously deduce the dynamics of the defaultable bond price process in terms of a Markov modulated stochastic differential equation. Then, by separating the utility maximization problem into the pre-default and post-default scenarios, we deduce two coupled Hamilton-Jacobi-Bellman equations for the post and pre-default optimal value functions and show a novel verification theorem for their solutions. We obtain explicit optimal investment strategies and value functions for an investor with logarithmic utility. We finish  with an economic analysis in the case of a market with two regimes and homogenous transition rates, and show the impact of the default intensities and loss rates on  the optimal strategies and value functions.

\vspace{0.3 cm}

\noindent{\textbf{AMS 2000 subject classifications}: 93E20, 60J20.}


\noindent{\textbf{Keywords and phrases}: Dynamic Portfolio Optimization, Credit Risk, Regime Switching Models, Utility Maximization,  Hamilton-Jacobi-Bellman Equations.}

\end{abstract}

\section{Introduction}
\footnote{This is an improved version of the original submission, where we fixed typos and updated references.}
Continuous time portfolio optimization problems are among the most widely studied problems in the field of mathematical finance.
Since the seminal work of \cite{Merton}, who explored stochastic optimal control techniques to provide a closed form solution to
the problem, a large volume of research has been done to extend Merton's paradigm to other frameworks and portfolio optimization problems (see, e.g. \cite{karatzas},
\cite{karatSh}, and \cite{Fleming}). Most of the models {proposed} in the literature {rely} on the assumption that the uncertainty in the asset price dynamics is governed by a {continuous process}, which is typically chosen to be a Brownian motion. In recent years, there has been an increasing interest in the use of regime switching model to capture the macro-economic regimes affecting the behavior of the market. More specifically,
the price of {the} security evolves with a different dynamics, typically identified by the drift and the diffusion coefficient associated to the {macro-economic} regime in place. Although regime switching models {arguably {are able to} incorporate a realistic description} of market behavior, they pose challenges in the context of pricing because they {lead to an incomplete market}, as the regime uncertainty cannot be hedged {away}.

In the {context of option pricing}, \cite{GuoPr} {studies} a two-regime switching {model}, where regimes {represent} the amount of information available to the market. \cite{BuffingtonEle} and \cite{BuffingtonEla} price European and American options under regime switching models, while \cite{elliottet} address the specification of an appropriate pricing martingale measure. \cite{GuoZhang} provide an explicit optimal stopping rule when pricing perpetual American put options, while \cite{GuoMiao} consider the relation between regime shifts and investment decision in the context of real options. \cite{RogersGraz} give a methodology to price barrier options with regime switching dividend process, while \cite{GapJean} obtain closed form expressions for European claims, assuming geometric Brownian motion dynamics with regime switching drifts.

Utility maximization problems under {regime switching} have been investigated in \cite{Cadenillas}, who consider the infinite horizon problem of maximizing the expected utility from consumption and terminal wealth in a market consisting of multiple stocks and a money market account, where both short rate and stock diffusion parameters evolve according to Markov-Chain modulated dynamics.
Similarly, \cite{Zari} considers an infinite horizon investment-consumption model where the agent can consume and distribute her wealth across a risk-free bond and a stock. \cite{NagaiRung} consider a finite horizon portfolio optimization problem for a risk averse investor with power utility, assuming that the coefficients of the risky assets in the economy are nonlinearly dependent on the Markov-chain modulated economic factors. \cite{KornKr} relax the assumption of constant interest rate and derive expressions for the optimal percentage of wealth invested in the money market account and stock, under the assumption of a diffusive short rate process with deterministic drift and constant volatility.

Most of the research done on continuous time portfolio optimization has concentrated on markets consisting of a risk-free asset, and of securities which only bear market risk. {These models do} not take into account securities carrying default risk, such as corporate bonds, even though the latter represent a significant portion of the market, comparable to the total capitalization of all publicly traded companies in the United States. In recent years, portfolio optimization problems have started to {incorporate} defaultable securities, but {assuming that the risky factors are modeled by continuous processes and more specifically by {Brownian} It\^o processes}. \cite{BielJang}
derive optimal finite horizon investment strategies for an investor with CRRA utility function, who optimally allocates her wealth among a defaultable bond, risk-free account, and stock, assuming constant interest rate, drift, volatility, and default intensity.
\cite{BoWang} consider an infinite horizon portfolio optimization problem, where an investor with logarithmic utility can choose a consumption rate, and invest her wealth across a defaultable perpetual bond, a stock, and a money market. They assume that both the historical intensity and the default premium process depend on a common Brownian factor. Unlike \cite{BielJang}, where the dynamics of the defaultable bond price process was derived from the arbitrage-free bond prices, \cite{BoWang} postulates the dynamics of the {defaultable} bond prices partially based on heuristic arguments.
\cite{Lakner} analyze the optimal investment strategy in a market consisting of a defaultable (corporate) bond and a money market account under a continuous time model, where bond prices can jump, and employ duality theory to obtain the optimal strategy. \cite{CallegaroJeanb} consider a market model consisting of several defaultable assets, which evolve according to discrete dynamics depending on partially observed exogenous factor processes. \cite{JiaoPham} combine duality and dynamic programming to optimize the utility of an investor with CRRA utility function,  in a market consisting of a riskless bond and a stock subject to counterparty risk. \cite{Bielgame} develop a variational inequality approach to pricing and hedging of a defaultable game option under a Markov modulated default intensity framework.

In this paper, we consider for the first time finite horizon dynamic portfolio optimization problems in defaultable markets with regime switching dynamics. We provide a general framework and explicit results on optimal value functions and investment strategies in a market consisting of a money market, a stock, and a defaultable bond.  Similarly to \cite{Cadenillas}, we allow
the short rate  and the drift and volatility of the risky stock to be all regime dependent. For the defaultable bond, we follow the reduced form approach to credit risk, where the global market information, including default, is modeled by the progressive enlargement of a reference filtration representing the default-free information, and the default time is a totally inaccessible stopping time with respect to the enlarged filtration, but not with respect to the reference filtration. {We also make the default intensities and loss given default rates to be all regime dependent}. The use of regime switching models for pricing defaultable bonds has proven to be very flexible when fitting the empirical credit spreads curve of corporate bonds as {illustrated in, e.g.,} \cite{JLT}, where the underlying Markov chain models credit ratings. 

Our main contributions are discussed next. First, we rigorously derive the dynamics of the defaultable bond under the historical measure from the price process, defined as a risk neutral {expectation}. 
Secondly, after separating the utility maximization problem into a pre-default and post-default dynamic optimization problem, we give and prove verification theorems for both {subproblems}. 
We show that the regime dependent pre-default optimal value function and bond investment strategy may be obtained as the solution of a coupled system of nonlinear partial differential equations (satisfied by the pre-default value function) and nonlinear equations (satisfied by the bond investment strategy), each corresponding to a different regime. Moreover, we obtain {the} interesting feature that the pre-default optimal value function and bond investment strategy depend on the corresponding regime dependent post-default value function.
Thirdly, we demonstrate our framework on the concrete case of an investor with logarithmic utility, and, show that both the optimal pre-default and post-default value functions amount to solving a system of ordinary linear differential equations, while the optimal bond strategy may be recovered as the unique solution of a decoupled system of equations, one for each regime. Under a two-regime market with homogenous transition rates, we are able to obtain explicit formulas, and illustrate the impact of default risk on the bond investment strategy and optimal value functions.

The rest of the paper is organized as follows. Section \ref{sec:model} introduces the market model. Section \ref{BPDSect} derives the dynamics of the defaultable bond under the historical measure, starting from the {risk-neutral bond price process}. Section \ref{sec:dynopt}
formulates the dynamic optimization problem. Section \ref{Sect:VerThr} gives and proves the two verification theorems associated to the post-default and pre-default case. Section \ref{sec:logappl} specializes the theorems given earlier to the case of an investor with logarithmic utility. Section \ref{sec:conclusion} {summarizes the main conclusions of} the paper. {The proofs of the main theorems and necessary lemmas are deferred to the appendix}.

\section{The Model}\label{sec:model}
Assume {$(\Omega,\G,\Gx,\Px)$} is a complete probability space, where $\Px$ is the real world probability measure  {(also called historical probability)}, ${\Gx}=(\gt)$ is an enlarged filtration given by
$\gt = \mathcal{F}_t \vee \mathcal{H}_t$ (the filtrations $\mathcal{F}_t$ and $\mathcal{H}_t$ will be introduced later).
 Let
$\{W_t\}$ be a standard Brownian motion on {$(\Omega,\G,\Fx,\Px)$, where $\Fx:=(\F_{t})_{t}$ is a suitable filtration satisfying the usual hypotheses of completeness and right continuity.}
We {also} assume that the \emph{states of the economy} are modeled by a continuous-time Markov process $\{X_t\}$ {defined} on $(\Omega,{\G,\Fx},\Px)$ with a finite state space $\{x_1,x_2,\ldots,x_N\}$.
Without loss of generality, we can identify the state space of $\{X_t\}$ to be a finite set of unit vectors $\{e_1, e_2, \ldots, e_N \}$, where {$e_i = (0,...,1,...0)^{'} \in \mathds{R}^N$ and $'$ denotes the transpose}.
We also assume that $\{X_t \}$ and $\{W_t\}$ are independent. Define $A(t)$ to be the Markov chain transition matrix $[a_{i,j}(t)]_{i,j=1,2,\ldots,N}$ (also referred to as the infinitesimal generator).
{The following semi-martingale representation is well-known} {(cf. \cite{elliottb}):}
\begin{equation}
X_t = X_0 + \int_0^t A^{\prime}(s) X_s ds + {M^{\Px}(t)}{,}
\label{eq:MCsemim}
\end{equation}
where {$M^{\Px}(t)=(M^{\Px}_{1}(t),\dots,M_{N}^{\Px}(t))'$} is a $\mathds{R}^N$-valued martingale process under $\Px$, and
$A(t) = [a_{i,j}(t)]_{i,j=1,\dots,N}$ is the so-called generator of the Markov process. Specifically, denoting $p_{i,j}(t,s) := \Px(X_s = j | X_t = i)$, for $s \geq t$, and $\delta_{i,j}={\bf 1}_{i=j}$, we have that
\[
	a_{i,j}(t) = \lim_{h \rightarrow 0} \frac{p_{i,j}(t,t+h) - \delta_{i,j}}{h};
\]
cf.  \cite{bielecki01}. In particular, $a_{i,i}(t) := -\sum_{j \neq i} a_{i,j}(t)$.

We consider a {frictionless} financial market  consisting of three instruments: a risk-free bank account, a defaultable bond, and a stock. The dynamics of each of the following {instruments} will depend on the underlying states of the economy {as follows}:

\medskip
\noindent\textbf{Risk-free bank account.} The instantaneous market interest rate at time $t$ is $r_t:= r(t,X_{t}):=\left<r,X_t\right>$, where $\left<\cdot,\cdot\right>$ denotes the standard inner product in $\mathbb{R}^{N}$ and $r=(r_1,r_2,\ldots,r_N)'$ are positive constants. This means that,
depending on the state of the economy, the interest rate $r_t$ will be different; i.e., if $X_{t}=e_{i}$ then $r_{t}=r_{i}$.
The dynamics of the price process $\{B_t\}$ which describes the risk-free bank account is given by
\begin{equation}
	{dB_t = r_t B_t d t.}
\end{equation}
\noindent\textbf{Stock price.}
We assume that the stock appreciation rate $\{\mu_t\}$ and the volatility $\{\sigma_t\}$ of the stock $S_t$ also depend on the {economy} regime in place $X_t$ in the following way:
\begin{equation}
\mu_t:= \mu(t,X_t) := \left<\mu, X_t\right>, \qquad \sigma_t:= \sigma(t,X_t) := \left<\sigma,X_t\right>
\end{equation}
where $\mu = (\mu_1,\mu_2,\ldots,\mu_N)'$ and $\sigma =  (\sigma_1,\sigma_2,\ldots,\sigma_N)'$ are constants denoting, respectively, the values of drift and volatility which can be taken depending on the different
economic regimes.
Hence, we assume that
\begin{equation}
dS_t = \mu_t S_t dt + \sigma_t S_t dW_t, \qquad S_0 = s.
\end{equation}
\noindent\textbf{Risky Bond price.} Unlike the previous two securities, where we have written
directly the dynamics under the historical measure, here we need to infer the historical dynamics (i.e. dynamics under the actual probability measure $\Px$) from the bond price {process, which is originally defined
under a suitably chosen risk-neutral pricing} measure $\Qx$.
Before {defining} the bond price, we need to {introduce} a default process. Let $\tau$ be a nonnegative random variable, {defined on $(\Omega,\G,\Px)$}, representing the default time of the counterparty selling the bond.
Let $\mathcal{H}_t = \sigma(H(u): u \leq t)$ be the filtration {generated by the \emph{default process} $H(t):= \idc_{\tau \leq t}$}, after
completion and regularization on the right, and also let   {$\Gx:=(\G_t)_{t}$} be the filtration
$\gt := \mathcal{F}_t \vee \mathcal{H}_t$.

 We use the canonical construction of the default time $\tau$ {in terms of a given hazard process {$\{h_t\}_{t\geq{}0}$, which will also be assumed to be driven by the Markov process $X$. Specifically,  throughout the paper we assume that {$h_t:=\left<h, X_t\right>$}, where  $h := (h_1,h_2,\ldots,h_N)'$} are positive constants. For future reference, we now give the details of the construction of the random time $\tau$.  {We assume the existence of an exponential random variable $\chi$ defined on the probability space $(\Omega,\F, \Px)$, independent of the process $(X_t)_{t}$}.
We define $\tau$ by setting
\begin{equation}
{\tau:=}\inf\{ t \in \mathds{R}^+ : \int_0^t {h_u} du \geq \chi \}.
\label{eq:taudef}
\end{equation}
It can be proven that {{$(h_t)_{t}$}} is the {$(\Fx,\Gx)$}-\emph{hazard rate} of $\tau$ {(see \cite{bielecki01}, Section 6.5 for details)}. That is, {$(h_t)_{t}$} is such that
\begin{equation}\label{MrtRprDftPrcP}
	{\xi^{\Px}_t} := H(t) - \int_0^t (1-H(u^{-})) h_u du
\end{equation}
is a ${\Gx}$-martingale under $\Px$, where $H(u^{-}) = \lim_{s\uparrow u} H(s) = \idc_{\tau < u}$.

An {important consequence of the previous construction is the following property.}
Let us fix $t>0$ and {$\F_{\infty}=\bigvee_{s\geq{}0} \F_{s}$}. For any $u \in \mathds{R}^+$, we have $\Px(\tau \leq u | \mathcal{F}_{\infty}) = 1-e^{-\int_0^u h_s ds}$. Therefore, for any $u\leq{}t$,
\[
 {\Px(\tau \leq {u} | \mathcal{F}_t) = \Ex^{\Px}\left[\Px(\tau \leq {u} | \mathcal{F}_{\infty}) | \mathcal{F}_t \right] = 1 - e^{-\int_0^{{u}} h_s ds} = \Px(\tau \leq {u} | \mathcal{F}_{\infty})}.
\]
 Plugging
$u=t$ inside the above expression, we obtain
\begin{equation}\label{eq:Hcondeq}
	{\Px(\tau \leq t | \mathcal{F}_t) = \Px(\tau \leq t | \mathcal{F}_{\infty})}.
\end{equation}
It was proven in \cite{BremaudYor} that Eq.~(\ref{eq:Hcondeq}) is equivalent to saying that any ${\Fx}$-square integrable martingale is also a $\Gx$-square integrable martingale. The latter {property} is also referred to as the $H$ hypothesis, and we will make use of this property later on.

The final ingredient in the bond pricing formula is the recovery process {{$(z_t)_{t}$}},  an ${\Fx}$-adapted} right-continuous with left-limits process to be fully specified below. Then, the time-$t$ price of the risky bond with maturity $T$ is given  by
\begin{equation}
p(t,T) := \Ex^{\Qx} \left[\int_t^T e^{-\int_t^{u} r_s ds} {z_u} {dH(u)}  + e^{-\int_t^T r_s ds} (1-H(T)) \bigg| \gt \right],
\label{eq:bondpr}
\end{equation}
where $\Qx$ is the {equivalent} risk-neutral measure used in pricing. Furthermore, we  adopt a pricing measure $\Qx$ such that, under $\Qx$,  $W$ is still a standard Wiener process and $X$ is a continuous-time Markov process (independent of $W$) with possibly different generator $A^{\Qx}(t):=[a_{i,j}^{\Qx}(t)]_{i,j=1,2,\ldots,N}$.

The existence of the {measure $\Qx$ in the previous paragraph} follows from the theory of change of measures for denumerable Markov processes (see, e.g., Section 11.2 in \cite{bielecki01}). Concretely, for $i\neq{}j$ and some bounded measurable functions  $\kappa_{i,j}:\mathbb{R}_+\to (-1,\infty)$, define
\begin{equation}\label{DfnRNGen}
	a_{i,j}^{\Qx}(t):=a_{i,j}(t) (1+\kappa_{i,j}(t)),
\end{equation}
and for $i=j$, define
\[
	a_{i,i}^{\Qx}(t):=-\sum_{k=1,k\neq{}i}^{N}a_{i,k}^{\Qx}(t).
\]
We also fix $\kappa_{i,i}(t)=0$ for $i=1,\dots,N$.
Now, consider the processes
\begin{equation}\label{CrsMrt}
	M_{t}^{i,j}:=H^{i,j}_{t}-\int_{0}^{t} a_{i,j}(u)
	H^{i}_{u}du,
\end{equation}
where
\begin{equation}\label{JmpTrnPrc1}
	H_{t}^{i}:= \idc_{\{X_{t}=e_i\}}, \quad \text{and}\quad
	 H_{t}^{i,j}:=\sum_{0<u\leq{}t}\idc_{\{X_{u^{-}}=e_{i}\}}\idc_{\{X_{u}=e_{j}\}}, \quad (i\neq{}j).
\end{equation}
The process $(M^{i,j}_{t})_{t}$ is known to be  a ${\Fx}$-martingale for any $i\neq{}j$ (see Lemma 11.2.3 in \cite{bielecki01}) and, since the $H$-hypothesis holds in our default framework, they are also ${\Gx}$-martingales.
Then, by virtue of Proposition 11.2.3 in \cite{bielecki01}, the probability measure $\Qx$ {on {$\Gx=(\mathcal{G}_{t})_{t}$}} with Radon-Nikod\'yn density $\{\eta_{t}\}$ given by
\begin{equation}\label{DfnDnsty}
	\eta_{t}=1+\int_{(0,t]} \sum_{i,j=1}^{N} \eta_{u^{-}} \kappa_{i,j}(u) d M_{u}^{i,j},
\end{equation}
is such that $X$ is a Markov process under $\Qx$ with generator $[a_{i,j}^{\Qx}(t)]_{i,j=1,2,\ldots,N}$. Without loss of generality, $\Qx$ can be taken such that $W$ is still a Wiener process independent of $X$ under $\Qx$.
\begin{equation}
X_t = X_0 + \int_0^t A^{\Qx}(s)' X_s ds + M^{\Qx}(t),
\label{eq:MCsemimb}
\end{equation}
where $M^{\Qx}$ is a $\mathds{R}^N$-valued martingale under $\Qx$. In particular, note that
\begin{equation}\label{RMM0}
	M^{\mathbb{Q}}(t)= M^{\mathbb{P}}(t)+\int_{0}^{t} ( A(s)'-A^{\mathbb{Q}}(s)')X_{s}ds.
\end{equation}
We emphasize that the distribution of the hazard rate process $h_{t}=\left<h,X_{t}\right>$ under the risk-neutral measure is different from that under the historical measure. Therefore, our framework allows modeling the default risk premium, defined as the ratio between risk-neutral and historical intensity, through the change of measure of the underlying Markov chain.

\section{{Defaultable bond price  dynamics}}\label{BPDSect}
We now proceed to obtain the bond price dynamics under both the risk-neutral and historical probability measures.
Eq.~(\ref{eq:bondpr}) may be rewritten as
\begin{align} \label{eq:bondprdyn}
p(t,T) &= \idc_{\tau>t} \Ex^{\Qx} \left[\int_t^T e^{-\int_t^u (r_s + {h_s}) ds} z(u) {h_u} du \bigg| \mathcal{F}_t \right] \\
&\quad+ \idc_{\tau>t}
\Ex^{\Qx} \left[e^{-\int_t^T (r_s + {h_s}) ds} \bigg| \mathcal{F}_t \right]\nonumber
\end{align}
which follows from Eq.~(\ref{MrtRprDftPrcP}), along with application of the following classical identity
$$\Ex^{\Qx}\left[\idc_{\tau > s} Y \bigg| \gt \right] = \idc_{\tau>t} \Ex^{\Qx}\left[e^{-\int_t^s {h_u} du} Y \bigg| \mathcal{F}_t \right],$$
where $t \leq s$ and $Y$ is a $\mathcal{F}_s$-measurable random variable (see \cite{bielecki01}, Corollary 5.1.1, for its proof).

We assume the recovery-of-market value assumption, i.e. {$z_t := (1-L_t) p(t^{-},T)$}, where $L_t$ is ${\Fx}$-predictable.
{As with the other factors in our model, we shall assume that {$L_t$} is of the form {$L_t:=\left<L,X_{t}\right>$} for some constant vector $L:=(L_{1},\dots,L_{N})' \in\mathbb{R}^{N}$. Under  the recovery-of-market value assumption, it follows
using a result in \cite{ds}, Theorem 1, that
\begin{equation}
p(t,T) = {\idc_{\tau>t} \Ex^{\Qx} \left[e^{-\int_t^T (r_s + {h_s} {L_s})ds}  \bigg| \mathcal{F}_t \right]}.
\label{eq:bondpricern}
\end{equation}

The following result gives the dynamics of the {defaultable} bond price {process} under the risk-neutral measure $\Qx$.
\begin{theorem}\label{Prop:DynBndQ}
Suppose that, for any $i\neq{}j$, the function $a_{i,j}^{\Qx}$ defined in (\ref{DfnRNGen}) is continuously differentiable in $(0,T)$ such that
	\begin{equation}\label{NdCndDf1}
		0<\inf_{s\in[0,T]}|a_{i,j}^{\Qx}(s)|\leq{}
		\sup_{s\in[0,T]}|a_{i,j}^{\Qx}(s)|<\infty \quad \& \quad \sup_{s\in(0,T)}\left|\frac{a_{i,j}^{\Qx}(s)}{ds}\right|<\infty.
	\end{equation}
Then, the {pre-default} dynamics of the bond price $p(t,T)$ under the risk-neutral measure $\Qx$ is given by
\begin{equation}\label{DynQ}
	dp(t,T)=p(t^{-},T)\left\{ \left[r_t + h_t (L_t-1)\right] dt + \frac{\left<{\psi(t)}, dM^{\Qx}(t)\right>}{\left<\psi(t), X_{t^{-}}\right>}- d\xi_t^{\Qx}\right\}
\end{equation}
where $(M^{\Qx}(t))_{t}$ is {the $N$-dimensional $({\Fx}, \Qx)$-martingale defined in (\ref{eq:MCsemimb})}, $(\xi_t^{\Qx})_{t}$ is a $({\Gx}, \Qx)$-martingale, and $\psi(t) := (\psi_1(t), \ldots,\psi_N(t))'$ is given by
\begin{equation}\label{OrgDfnPsi}
	\psi_i(t) := \Ex^{\Qx} \left[\left.e^{-\int_t^{T} (r_s + h_s L_s) ds} \right| {X_t = e_i} \right]. \\
\end{equation}
\end{theorem}

The proof is reported in Appendix \ref{app:dynamicsbond}.
We also have the following dynamics under the historical probability measure.
\begin{proposition}\label{Prop:DynBndP}
	Under the assumptions of Theorem \ref{Prop:DynBndQ}, the {pre-default} dynamics of the bond price $p(t,T)$ under the historical measure {$\Px$ is} given by
\begin{equation}\label{DRWP0}
dp(t,T)=p(t^{-},T)\left\{ \left[r_t + h_t (L_t-1)+{D(t)} \right] dt + \frac{\left<{\psi(t)}, dM^{\Px}(t)\right>}{\left<\psi(t), X_{t^{-}}\right>}-d\xi^{\Px}_{t}\right\}.
\end{equation}
where $(M^{\Px}(t))_{t}$ is {the $N$-dimensional $({\Fx}, \Px)$  martingale defined in (\ref{eq:MCsemim})}, $(\xi_t^{\Px})_{t}$ is {the $({\Gx}, \Px)$-martingale defined in (\ref{MrtRprDftPrcP})}, and $D(t):=\left< (D_{1}(t),\dots,D_{N}(t))',X_{t}\right>$ with
\begin{equation}\label{NDfnDf}
	D_{i}(t):=\sum_{j=1}^{N} (a_{i,j}(t)-a_{i,j}^{\Qx}(t))\frac{\psi_{j}(t)}{\psi_{i}(t)}
	= \sum_{j\neq i} (a_{i,j}(t)-a_{i,j}^{\Qx}(t))\left(\frac{\psi_{j}(t)}{\psi_{i}(t)}-1\right).
\end{equation}
\end{proposition}
	

\section{Optimal Portfolio Problem}\label{sec:dynopt}
 We consider an investor who wants to maximize her wealth at time $R \leq T$ by dynamically allocating her financial wealth into (1) a risk-free bank account, (2) a risky asset, {and} (3) defaultable bond. The investor does not have intermediate consumption nor capital income to support her purchase of financial {assets.} 

 Let us denote by $\nu_t^B$ the number of shares of the risk-free bank account that the investor buys (${\nu}_t^B>0$) or sells (${\nu}_t^B<0$) at time $t$. Similarly, ${\nu}_t^S$ and  {${\nu}_t^P$} denote the investor's portfolio positions on the stock and risky bond at time $t$, respectively.
 The process $({\nu}_t^B,{\nu}_t^S,{\nu}_t^P)$ is called a \emph{portfolio process}. We {denote $V_t({\nu})$} the wealth of the portfolio process ${\nu}=({\nu}^B,{\nu}^S,{\nu}^P)$ at time $t$, i.e.
 $$
 V_t({\nu}) = {\nu}_t^B B_t + {\nu}_t^S S_t + {\nu}_t^P p(t,T).
 $$
As usual, we require the processes ${\nu}_t^B,{\nu}_t^S$, and ${\nu}_t^P$ to be {${\Fx}$-predictable. }
We also assume  {the following {self-financing} condition}:
\[
	d V_{t}={\nu}_{t}^{B} d B_{t}+{\nu}^{S}_{t} dS_{t}+{\nu}_{t}^{P} d p(t,T).
\]
Given an initial state configuration $(x,z,v) \in \mathbb{E}:=\{e_1,e_2,\ldots,e_N \} \times \{0,1\} \times (0,\infty)$, we define {the class of admissible strategies} $\mathcal{A}:=\mathcal{A}(v,i,z)$ to be a set of {(self-financing)} portfolio processes ${\nu}$ such that {$V_{t}({\nu})\geq {}0$} for all $t\geq{}0$ when
$X_{0}=x$, $H_{0}=z$, and $V_{0}=v$.

Let $\pi:=(\pi_t^B,\pi_t^S,\pi_t^{P})$ be defined as
 \begin{equation}
 \pi_t^B := \frac{{\nu}_t^B B_t}{V_{t-}({\nu})}, \quad
\pi_t^S := \frac{{\nu}_t^S S_t}{V_{t-}({\nu})},
\quad
 \pi_t^P=  \idc_{\tau > t} \frac{{\nu}_t^{P} {p(t^{-},T)}}{V_{t-}({\nu})},
  \end{equation}
if $V_{t-}({\nu})>0$, while $\pi_t^B=\pi_t^P=\pi_t^S=0$, when $V_{t-}({\nu})=0$.
The vector $\pi:=(\pi_t^B,\pi_t^S,\pi_t^P)$, called a \emph{trading strategy}, {represents} the corresponding fractions of wealth invested in each asset at time $t$. Note that if $\pi$ is admissible, then the dynamics of the resulting wealth process $V^{\pi}$ {can be written as}
\[
 	d V_{t}^{\pi}=V^{\pi}_{t^{-}}\left\{\pi_{t}^{B}\, \frac{d B_{t}}{B_{t}}+\pi^{S}_{t}\, \frac{dS_{t}}{S_{t}}+{\pi_{t}^{P}}\, {\frac{d p(t,T)}{p(t^{-},T)}}\right\},
\]
under the convention that $0/0=0$. This convention is needed to deal with the case when default has occurred ($t>\tau$), so that $p(t^{-},T)$=0 and we fix $\pi^{P}_{t}=0$.
Using the dynamics derived in Proposition \ref{Prop:DynBndP} {and that $\pi^{B}+\pi^{P}+\pi^{S}=1$}, we have the following dynamics of the wealth process
\begin{eqnarray}
\nonumber dV_t^\pi &=&
V_{t^{-}}^\pi \bigg[\left\{r_{t}+ \pi_t^S (\mu_{t}-r_{t})+\pi_{t}^{P}[h_{t}(L_{t}-1)+D(t)] \right\}dt \\
 & &  \qquad +{\pi_t^S} \sigma_t d W_{t} +\pi_{t}^{P}  \frac{\left<\psi(t), d M^{\Px}(t)\right>}{\left<\psi(t),X_{t^{-}}\right>}
- \pi_{t}^{P}d\xi_{t}^{\Px} \bigg],
\label{eq:wealtheqsimpl}
\end{eqnarray}
under the actual probability $\Px$.

\subsection{{The utility maximization problem}}
For an initial value {$(x,z,v) \in \mathbb{E}$ and an admissible strategy $\pi=(\pi^B,\pi^S,\pi^p)\in\mathcal{A}(x,z,v)$,} 
 let us define the objective functional to be
\begin{equation}
{J_{R}(x,z,v; \pi)} := \Ex^{\Px}\left[U(V_R^{\pi}) \bigg| X_0=x, H_0=z, V_0=v \right];
\end{equation}
i.e. we are assuming that the investor starts with $v$ dollars (its initial wealth), that the initial default state is $z$ ($z=0$ means that no default has occurred yet),
and the initial value for the underlying state of the economy is $x$. The constraint $V_0=v$ is also called the budget constraint.
{As usual, we assume that the utility function $U:[0,\infty)\to\mathbb{R}\cup \{\infty\}$ is strictly increasing and concave.}

Our goal is to maximize the objective functional {$J(x,z,v;\pi)$} for a suitable class of admissible strategies {$\pi_{t}:=(\pi_{t}^{B},\pi_{t}^{S},\pi_{t}^{P})$. Furthermore, we shall focus on feedback or Markov strategies of the form}
\[
	{\pi_{t} = (\pi^B_{_{C_{t^{-}}}}(t,V_{t^{-}},H(t^{-}), \pi^S_{_{C_{t^{-}}}}(t,V_{t^{-}},H(t^{-}),\pi^P_{_{C_{t^{-}}}}(t,V_{t^{-}},H(t^{-})),}
\]
for {some functions {$\pi^{B}_{i},\pi^{P}_{i},\pi^{S}_{i}:[0,\infty)\times [0,\infty)\times \{0,1\}\to \mathbb{R}$} such that $\pi^B_{i}(t,v,z)+\pi^S_{i}(t,v,z)+\pi^{P}_{i}(t,v,z)=1$}.

As usual, we consider instead the following {dynamical} optimization problem:
\begin{equation}\label{VFG}
	{\varphi^R(t,v,i,z) := \sup_{{\pi \in \mathcal{A}_{t}(v,i,z)}} \Ex^{\Px}\left[ U(V_R^{\pi,t,v}) \bigg| {V_t = v}, X_t={e_i}, H(t) =z \right],}
\end{equation}
for each $(v,i,z) \in  {(0,\infty) \times \{1,2,\ldots,N \} \times \{0,1\}}$,
where
\begin{align}
\nonumber dV_s^{\pi,t,v} &=
V_{s^{-}}^{\pi,t,v} \bigg[\left\{r_{s}+ \pi_s^S (\mu_{s}-r_{s})+\pi_{s}^{P}{(1-H(s^{-}))}[h_{s}(L_{s}-1)+D(s)] \right\}ds\\
\nonumber &+{\pi_s^S} \sigma_s d W_{s} +\pi_{s}^{P} {(1-H(s^{-}))}
\frac{\left<\psi(s), d M^{\Px}(s)\right>}{\left<\psi(s),X_{s^{-}}\right>}
- \pi_{s}^{P}d\xi_{s}^{\Px} \bigg], \quad   {s\in[t,R],}
\\
V_{t}^{\pi,t,v}&=v.
\label{eq:wealtheqsimpl2}
\end{align}
The class of processes $\mathcal{A}_{t}(v,i,z)$  is defined as follows:
\begin{definition}
	Throughout, $\mathcal{A}_{t}(v,i,z)$ denotes a {suitable} class of ${\Fx}$-predictable locally bounded {feedback} trading strategies
	\[
		 {\pi_s:=(\pi_{s}^{S},\pi_{s}^{P}):=(\pi_{C_{s^{-}}}^S(s,V_{s^{-}}^{\pi,t,v},H(s^{-})),\pi_{C_{s^{-}}}^P(s,V_{s^{-}}^{\pi,t,v},H(s^{-})))}, \quad s\in[t,R],
	\]
	such that (\ref{eq:wealtheqsimpl2}) admits a unique strong solution ${\{V_{s}^{\pi,t,v}\}_{s\in[t,R]}}$ and
$V_{s}^{\pi,t,v}>{}0$ for any $s\in[t,R]$ when $X_{t}=e_{i}$ and $H(t)=z$. {Throughout this paper}, a trading strategy satisfying these conditions is simply said to be $t$-admissible (with respect to the initial conditions $V_{t}=v$, $X_{t}=e_{i}$, and $H_{t}=z$).
\end{definition}

\begin{remark}
	As it will be discussed below (see Eqs.~(\ref{GenDynd0}), (\ref{CffDynWlth}), and (\ref{JmpWlthPrc})), the jump $\Delta V_{s}:=V_{s}-V_{s^{-}}$ of the process (\ref{eq:wealtheqsimpl2}) at time $s$ is given by
	\begin{equation}\label{JmpWlthPrc0}
	\Delta V_{s}=V_{s^{-}}\bigg(\sum_{i=1}^{N}\sum_{j\neq{}i}{\pi_{i}^{P}(s,V_{s^{-}},H(s^{-}))}{\frac{\psi_{j}(s)-\psi_{i}(s)}{\psi_{i}(s)}}\Delta H_{s}^{i,j}
	- \pi_{s}^{P}\Delta H(s)\bigg).
\end{equation}
Since for $\{V_{s}\}_{s\geq{}t}$ to be strictly positive, it is necessary and sufficient that $\Delta V_{s}>-V_{s^{-}}$ for {any} $s>{}t$ {a.s. (cf.  \cite[Theorem 4.61]{Jacod})}, we conclude that in order for $\pi^{P}$ to be admissible, it is necessary that,
\begin{equation}\label{NcsCndAdm}	
	M_{i}:={\max_{j \neq i:\psi_{i}(s)<\psi_{j}(s)} \left( -\frac{\psi_i(s)}{\psi_j(s)-\psi_{i}(s)}\right)} < \pi_i^P(s,v,z)  < 1,
\end{equation}
{for any  $s,v>0$, $z\in\{0,1\}$, and $i=1,\dots,N$}, {where we set $M_{i}:=-\infty$ if {$\psi_{i}(s)\geq \psi_{j}(s)$} for all $j\neq{}i$}.
\end{remark}

\section{Verification Theorems}\label{Sect:VerThr}
As it is usually the case, we start by deriving the HJB formulation of the value function (\ref{VFG}) via heuristic arguments. We then verify that the solution of the proposed HJB equation (when it exists and satisfies other regularity conditions) is indeed optimal (the so-called verification theorem). Let us assume for now that $\varphi^R(t,v,i,z)$ is $C^1$ in $t$ and $C^2$ in $v$  for each $i$ and $z$.
Then, using It\^o's rule along the lines of Appendix \ref{Sect:DrvGnr}, we have that
$$
\varphi^R(t, V_t^{\pi}, C_t, H(t)) = \varphi^R(r,V_r^{\pi},C_r,H_r) + \int_r^t \mathcal{L} \varphi^R(s,V_s^{\pi},C_s, H_s) ds +
 \mathcal{M}_t -\mathcal{M}_{r},
$$
where $(C_{t})_{t}$ is the Markov process defined in (\ref{DfnChnPrc}), $\mathcal{L}$ is the infinitesimal generator of $(t,V_{t},C_{t},H_{t})$ given in Eq.~(\ref{eq:generator}), and {$(\mathcal{M}_t)_{t}$} is the martingale given by
Eq.~(\ref{eq:lmt}). Next, if $r < t <R$, by virtue of the dynamic programming principle, we expect that
\begin{equation}\label{DnmPrgrm}
\varphi^R(r,V_r^{\pi}, C_r, H_r) = \max_{\pi} \Ex\left[ \varphi^R(t,V_t^{\pi},C_t,H(t)) | \mathcal{G}_{r} \right].
\end{equation}
Therefore, we obtain
\(
	\Ex\left[\left. \int_r^t \mathcal{L} \varphi^R(s,V_s^{\pi},C_s, H_s) ds\right|\mathcal{G}_{r} \right] \leq 0,
\)
with the inequality becoming an equality if $\pi=\widetilde{\pi}$, where $\widetilde{\pi}$ denotes
the optimum. Now, evaluating the derivative with respect to $t$, at $t = r$, we deduce the following HJB equation:
\begin{equation}\label{HJBE0}
	\max_{\pi} \mathcal{L} \varphi^R(r,v,i,z) = 0,
\end{equation}
with boundary condition
\(
\varphi^R(T,  v, i, z) = U(v).
\)

In order to further specify (\ref{HJBE0}), let us first note that the dynamics {(\ref{eq:wealtheqsimpl2})} can be written in the form
\begin{align}\label{GenDynd0}
 dV_s&= \alpha_{_{C_{s}}} ds+ \vartheta_{_{C_{s}}} d W_{s}
 +\sum_{j=1}^{N}\beta_{_{C_{s^{-}},j}} dM^{\Px}_{j}(s)
- \gamma_{_{C_{s^{-}}}} d\xi_{s}^{\Px},\quad (t<s<R),
\end{align}
with coefficients
\begin{align}
\beta_{i,j}(t,v,z) &= v {\pi_i^P(t,v,z)} (1-z)\frac{\psi_{j}(t)}{\psi_{i}(t)}, \quad {\gamma_{i}(t,v,z)} = v{\pi_{i}^{P}(t,v,z)}(1-z), \nonumber \\
\alpha_i(t,v,z) &=v \left[ r_{i}+{\pi_{i}^{S}}(\mu_{i}-r_{i})+{\pi^{P}_{i}}(1-z)(h_{i}(L_{i}-1) +{D_{i}(t)}) \right]\label{CffDynWlth} \\
\vartheta_i(t,v,z) &= {\pi^{S}_{i}(t,v,z)} \sigma_i v,\nonumber
\end{align}
where $D_{i}(t)$ is defined as in (\ref{NDfnDf}).
Using the expression for the generator in Eq.~(\ref{eq:generator}), the notation $\varphi_{i,z}(t,v):=\varphi^{R}(t,v,i,z)$, and the relationship $\pi^{B}=1-\pi^{S}-\pi^{P}$, (\ref{HJBE0}) can be written as follows {for each $i=1,\dots, N$}:
\begin{align}\nonumber
	 0&=\frac{\partial \varphi_{i,z}}{\partial t} +  v r_{i} \frac{\partial \varphi_{i,z}}{\partial v}+ z \sum_{j\neq{}i} a_{i,j}(t)\left[\varphi_{j,z}(t,v)-
	  \varphi_{i,z}(t,v)\right]\\ \nonumber
	 &\quad\quad+\max_{\pi^{S}_{i}}\left\{
	  \pi^{S}_{i}(\mu_{i}-r_{i})v \frac{\partial \varphi_{i,z}}{\partial v}+ (\pi^{S}_{i})^{2}\frac{\sigma_{i}^{2}}{2}v^{2}\frac{\partial^{2} \varphi_{i,z}}{\partial v^{2}}\right\}\\
	  \nonumber
	 &\quad \quad +(1-z)\max_{{\pi^{P}_{i}}}\bigg\{ {\pi^{P}_{i}}\theta_{i}(t) v\frac{\partial \varphi_{i,z}}{\partial v}
	  +h_{i} \left[\varphi_{i,1}(t,v(1-{\pi^{P}_{i}}))-\varphi_{i,z}(t,v)\right]	\\
	  &\quad\quad+\sum_{j\neq{}i} a_{i,j}(t)\left[\varphi_{j,z}\left(t,v\left[1+{\pi^{P}_{i}}{\left(\frac{\psi_{j}(t)}{\psi_{i}(t)}-1\right)}\right]\right)-
	  \varphi_{i,z}(t,v)\right]\bigg\},
	  \label{eq:HJBpdepre00}
\end{align}
where
\begin{equation}\label{DfnTheta}
	\theta_{i}(t):={h_{i}L_{i} -\sum_{j\neq{}i} a_{i,j}^{\Qx}(t)\left(\frac{\psi_{j}(t)}{\psi_{i}(t)}-1\right)}.
\end{equation}
We can consider two separate cases
\begin{equation}\label{PreDflt}
{\bar{\varphi}^R(t,v,i) =  \varphi_{i,0}(t,v)=\varphi^R(t,v,i,0),}  \qquad (\text{pre-default case})
\end{equation}
and
\begin{equation}\label{PstDflt}
{\underline{\varphi}^R(t,v,i) =  \varphi_{i,1}(t,v)=\varphi^R(t,v,i,1),} \qquad (\text{post-default case}).
\end{equation}

Section \ref{sec:postdefault} give a verification theorem for the post-default case, while Section \ref{sec:predefault} gives
a verification theorem for the pre-default case.

\subsection{Post-Default Case} \label{sec:postdefault}
In the post-default case, we have that $p(t,T) = 0$, for each $\tau < t \leq T$. Consequently, $\pi^P_t = 0$ for $\tau < t \leq T$ and, since $\pi_t^B = 1-\pi_t^S-\pi_t^P$, we can take $\pi = \pi^S$ as control.

Below, $\eta_i := \frac{\mu_i - r_i}{\sigma_i}$ denotes the sharpe ratio of the risky {asset} under the $i^{th}$ state of economy and $C^{1,2}_{0}$ denotes the class of functions $\varpi:[0,R] \times \mathbb{R}_{+}\times \{1,\dots,N\}\to \mathbb{R}_{+}$  such that
\[
	\varpi(\cdot,\cdot,i)\in C^{1,2}((0,R)\times\mathbb{R}_{+}) \cap C([0,R]\times \mathbb{R}_{+}),\quad  \varpi_{v} (s,v,i)\geq{}0,\quad \varpi_{vv} (s,v,i)\leq{}0,
\]
for each $i=1,\dots,N$. We have the following verification result, whose proof is reported in Appendix \ref{SectVerification}:
\begin{theorem}\label{MntPstDftVal}
Suppose that there exist a function ${\underline{w}}\in C^{1,2}_{0}$ that solve the
nonlinear Dirichlet problem
\begin{equation}
{\underline{w}}_t(s,v,i) -\frac{\eta_i^{2}}{2} \frac{{\underline{w}}_v^2(s,v,i)}{{\underline{w}}_{vv}(s,v,i)} + r_i v {\underline{w}}_v(s,v,i) + \sum_{j \neq i} a_{i,j}(s)
\left( {\underline{w}}(s,v,j) - {\underline{w}}(s,v,i) \right) = 0,
\label{eq:dirich}
\end{equation}
for any $s\in(0,R)$ and $i=1,\dots,N$, with terminal condition ${\underline{w}}(R,v,i) = U(v)$.
We assume additionally that {$\underline{w}$} satisfies
	\begin{equation}\label{KCUB}
		{\rm (i)}\;\;
		|{\underline{w}}(s,v,i)|\leq{} D(s)+E(s) v,\quad
		{\rm (ii)}\;\;\left|\frac{{\underline{w}}_{v}(s,v,i)}{{\underline{w}}_{vv}(s,v,i)}\right|\leq {G(s)} (1+v),
	\end{equation}
	for {some} locally bounded functions $D,E,G:\mathbb{R}_{+}\to \mathbb{R}_{+}$.
Then, the following statements hold true:
\begin{enumerate}
	\item[(1)] ${\underline{w}}(t,v,i)$ coincides with the optimal value function $\underline{\varphi}^R(t,v,i) = \varphi^R(t,v,i, 1)$  in (\ref{VFG}), when
	$\mathcal{A}_{t}(v,i,1)$ is constrained to the class of $t$-admissible feedback controls $\pi_{s}^{S}={\pi}_{_{C_{s}}}(s,V_{s})$ such that ${\pi_{i}}(\cdot,\cdot)\in C([0,R]\times\mathbb{R}_{+})$ for each $i=1,\dots,N$ and 	
	\begin{equation}\label{NCFUB}
		|v \pi_{i}(s,v)|\leq{} G(s)(1+v),
	\end{equation}
	for a locally bounded function $G$. If the solution ${\underline{w}}$ is non-negative, then condition (\ref{NCFUB}) is not needed.
	\item[(2)] The optimal feedback control $\{\pi_s^S\}_{s\in[t,R)}$, denoted by
$\widetilde{\pi}_s^S$, can be written as $\widetilde{\pi}_s^S = {\widetilde{\pi}_{_{C_{s}}}(s,V_{s})}$ with
\begin{equation}
	{\widetilde{\pi}_{i}(s,v)}=-{\frac{\eta_{i}}{\sigma_i}} \frac{{\underline{w}}_v(s,v,i)}{v {\underline{w}}_{vv}(s,v,i)}.
	\label{eq:optpi}
\end{equation}
\end{enumerate}
\end{theorem}

\subsection{Pre-Default Case}\label{sec:predefault}
In the pre-default case ($z=0$), we take $\pi^S$ and $\pi^{P}$ as our controls.
We then have the following verification result:
\begin{theorem}\label{MntPstDftVal2}
Suppose  that the conditions of Theorem \ref{MntPstDftVal} are satisfied {and, in particular, let $ \underline{w}\in C_{0}^{1,2}$ be the solution of (\ref{eq:dirich})}.
Assume that  $\bar{w}\in C^{1,2}_{0}$ and {$p_{i}=p_{i}(s,v)$, $i=1,\dots,N$,} solve simultaneously the following system of equations:
\begin{align}
	\nonumber
	&\theta_{i}(s) \bar{w}_v(s,v,i)- h_{i}  \underline{\varphi}^R_v(s,v(1-{p_{i}}),i)\\
	\label{eq:optpiP}
	&\quad \quad +  \sum_{j\neq{}i}
	a_{i,j}(s){\left( \frac{\psi_{j}(s)}{\psi_{i}(s)} -1 \right)}  \bar{w}_v\left(s,v\left[1+{p_{i}}{\left(\frac{\psi_{j}({s})}{\psi_{i}({s})}-1\right)}\right],j\right)=0,\\
	\nonumber
	& \bar{w}_t(s,v,i) -\frac{\eta_i^{2}}{2} \frac{\bar{w}_v^2(s,v,i)}{\bar{w}_{vv}(s,v,i)} + r_i v \bar{w}_v(s,v,i) \\
	\nonumber
	&\quad+\bigg\{ {p_{i}}\theta_{i}(t) v \bar{w}_v(s,v,i)
	  +h_{i} \left[{\underline{w}}(s,v(1-{p_{i}}),i)
	  - \bar{w}(s,v,i)\right]	\\
	  &\quad\quad+\sum_{j\neq{}i} a_{i,j}(t)\left[\bar{w}\left(s,v\left(1+{p_{i}}{\left(\frac{\psi_{j}(s)}{\psi_{i}(s)}-1\right)}\right),j\right)-
	  \bar{w}(s,v,i)\right]\bigg\} = 0,
	\label{eq:dirichpre}
\end{align}
for $t<s<R$, with terminal condition $\bar{w}(R,v,i) = U(v)$. We also assume that ${p_{i}(s,v)}$ satisfies  {(\ref{NcsCndAdm}) and (\ref{NCFUB})}  (uniformly in $v$ and $i$) {and $\bar{w}$ satisfies (\ref{KCUB})}.
Then, the following statements hold true:
\begin{enumerate}
	\item[(1)] $\bar{w}(t,v,i)$ coincides with the optimal value function $\bar{\varphi}^R(t,v,i) = \varphi^R(t,v,i, 0)$ in (\ref{VFG}), when
	$\mathcal{A}_{t}(v,i,0)$ is constrained to the class of $t$-admissible feedback controls {$(\pi_{s}^{S},\pi_{s}^{P})=({\pi}^{S}_{_{C_{s^{-}}}}(s,V_{s^{-}},H(s^{-})),{\pi}^{P}_{_{C_{s^{-}}}}(s,V_{s^{-}},H(s^{-})))$} such that
	\[
		{\pi_{i}^{S}(\cdot,\cdot,z),\pi^{P}_{i}(\cdot,\cdot,z)\in C([0,R]\times\mathbb{R}_{+})},
	\]
	for each $i=1,\dots,N$, $\pi^{S}$ satisfies (\ref{NCFUB}) for a locally bounded function $G$, and $\pi^{P}$ satisfies {(\ref{NcsCndAdm}) and (\ref{NCFUB})} (uniformly in {$v,i,z$}). If the solution {$\bar{w}$} is non-negative, then these bound conditions are not needed.
	\item[(2)]  The optimal feedback controls are given by
	{$\widetilde{\pi}_s^S:=\widetilde{\pi}^S_{_{C_{s^{-}}}}(s,V_{s},H(s))$ and $\widetilde{\pi}_s^P:=\widetilde{\pi}^P_{_{C_{s^{-}}}}(t,V_{t},H(s))$} with
\begin{eqnarray}
	{\widetilde{\pi}^{S}_{i}(s,v,z)}&=&{-\frac{\eta_{i}}{\sigma_i} \frac{\bar{w}_v(s,v,i)}{v \bar{w}_{vv}(s,v,i)}(1-z)
	-\frac{\eta_i}{\sigma_i} \frac{\underline{w}_v(s,v,i)}{v \underline{w}_{vv}(s,v,i)}z},\\
	{\widetilde{\pi}^{P}_{i}(s,v,z)}&=&{p_{i}}(s,v)(1-z).\label{OptPrDf}
\end{eqnarray}
\end{enumerate}
\end{theorem}

\section{Application to Logarithmic Utility}\label{sec:logappl}
The objective of this section is to specialize the framework developed in the previous sections to concrete choices of utility functions. We focus on the {logarithmic utility function} $U(v) = \log(v)$. The framework, however, {can} be applied to other concave utility functions, such as {power or negative exponential utilities}. 
{For the sake of clarity and conciseness, the details of the numerical implementation for other HARA functions is being deferred to the follow-up paper \cite{CapFig}}.

\subsection{Explicit Solutions}
{Let us recall that the investor's horizon $R$ is assumed to be less than the maturity $T$}  of the defaultable bond. 
Before proceeding, we state without proof a fundamental result from the theory of ordinary differential equations.

\begin{lemma}[\cite{Codd}]\label{lemmaode}
Suppose that the $n \times n$ matrix $F(t)$ and the $n \times 1$ vector $b(t)$, are both continuous on an interval $I \in \mathbb{R}$. Let
$t_0 \in I$. Then, for every choice of the vector $x_0$, we have that the system
\begin{eqnarray*}
{x'(t) = F(t) x(t) + b(t), \qquad x(t_0) = x_{0},}
\end{eqnarray*}
has a unique vector-valued solution $x(t)$ that is defined on the same interval $I$.
\end{lemma}

 We first give a lemma, which will be used later to derive {the} dynamics of the optimal value functions, and relations satisfied by the optimal investment strategies{.}

\begin{lemma}\label{lemma:pcontin}
The system of equations
\begin{equation}
 {\theta_i(s) - \frac{h_i}{1-p_{i}} + \sum_{j \neq i} a_{i,j}(s) \frac{{\psi_j(s)-\psi_{i}(s)}}{\psi_i(s) + p_{i}{( \psi_j(s)-\psi_{i}(s))}} = 0},
\label{eq:optimalpi}
\end{equation}
{for $i=1,\dots,N$}, admits a unique real solution {$p_{i}(s)$} in the interval $(M_i,1)$, where
{$M_i\in[-\infty,0)$ is defined as in (\ref{NcsCndAdm})}. Moreover, if for each ${i,j} = 1, \ldots, N$, $a_{i,j}$ and $a_{i,j}^{\Qx}$ are continuous {functions}, then {$p(s,i)$} is a continuous function of $s$.
\end{lemma}
\noindent The proof of this lemma is reported in Appendix \ref{SectExplConstr}. {The following is our main result in this section.}
\begin{proposition}\label{prop:logvaluefn}
Assume that {the} $a_{i,j}^{\Qx}$'s and $a_{i,j}$'s are continuous in $[0,T]$. Let
{$p(t) = [p_{1}(t),p_{2}(t), \ldots, p_{N}(t)]$} be the unique continuous solution in $[0,T]$ of the nonlinear system of equations~(\ref{eq:optimalpi}). Then, {the following statements hold true:}
\begin{enumerate}
\item[(1)] The optimal post-default value function is given by
$${\underline{\varphi}^R(t,v,i)} = \log(v) + K(t,i){,}$$
where $0 \leq t \leq R$, and $K(t) = [K(t,1), K(t,2), \ldots K(t,N)]$ is {the unique solution of} {the linear system of first order differential equations}
\begin{align}
\label{eq:postdeflog}
 &{K_t(t,i) + r_i + \frac{\eta_i^2}{2} + \sum_{j \neq i} a_{i,j}(t) K(t,j) + a_{i,i}(t) K(t,i) = 0},
\end{align}
{with terminal conditions $K(T,i) = 0$},
for $i=1,\dots,N$.
\item[(2)] The optimal percentage of wealth invested in stock is given by {$\tilde{\pi}^S(t) = [\tilde{\pi}^S_{1}(t), \tilde{\pi}^S_{2}(t), \ldots, \tilde{\pi}_{N}^S(t)]$}, where
${\tilde{\pi}^S_{j}(t)} = \frac{\mu_j - r_j}{\sigma_j^2}$, $ 0 \leq t \leq R$.
\item[(3)] The optimal percentage of wealth invested in the defaultable bond is {$\tilde{\pi}^P_{t}=p_{_{C_{t^{-}}}}(t)$}, while the optimal pre-default value function is
$${\overline{\varphi}^R(t,v,i)} = \log(v) + J(t,i),$$ where
$J(t) = [J(t,1), J(t,2), \ldots J(t,N)]$ is the unique solution of the  {linear system of first order differential} equations
\begin{align}
\nonumber &J_t(t,i) + \frac{\eta_i^2}{2} + r_i + {p_{i}(t)} \theta_i(t) + h_i \left(\log(1-{p_{i}(t)}) + K(t,i) - J(t,i) \right) +   \\
&\quad\sum_{j \neq i} a_{i,j}(t)
\left[ \log \left( 1 + {p_{i}(t)} {\left( \frac{\psi_j(t)}{\psi_i(t)} -1 \right) } \right) + J(t,j) - J(t,i) \right] = 0, \label{eq:predeflog1}
\end{align}
{with terminal conditions $J(R,i) =  0$}, for $i=1,\dots,N$.
\end{enumerate}
\end{proposition}
\noindent The proof of this proposition is reported in Appendix \ref{SectExplConstr}. {We notice that the only difference between the pre-default and post-default optimal value function lies in the time and regime component}.
Moreover, we obtain that the optimal proportion of wealth {invested} in stocks is constant in every economic regime, and does not depend on {the} time or on the current level of wealth. This is consistent with the findings in {\cite{Cadenillas} and \cite{BoWang}}.
We also find that the optimal proportion of wealth allocated on the defaultable bond is time and regime dependent, but independent on the current level of wealth. \cite{BoWang} find that the optimal allocation only depends on time through the default risk premium.

In the case where the rate matrix $A$ is homogenous {(i.e. $a_{i,j}(t) = a_{i,j}$)} and the number of regimes is $N = 2$,
{we} can obtain {closed} form expressions for the optimal pre-default value function, post-default value function{,} and optimal bond fraction {$\widetilde{\pi}^P$}. Let us introduce the following notation:
\begin{align*}
\zeta_i &= r_i + \frac{\eta_i^2}{2}, 
\quad \tilde\psi_{1}(t)= \frac{\psi_{2}(t)}{\psi_{1}(t)} - 1,
\quad \tilde\psi_{2}(t)=\frac{\psi_{1}(t)}{\psi_{2}(t)} - 1,\\
\cone &= h_1 -  a_{1,1},
 \cthree = h_2 - a_{2,2},
\quad \CPlus = \CPlusExpr,   \\
g_{i}(t) &=  {\tilde{\pi}^P_{i}(t)} \theta_i(t) + h_i\left( \log(1-\tilde{\pi}^P(t,i)) + K(t,i) \right)- a_{i,i} \log\left(1+ {\tilde{\pi}^P_{i}(t)} \tilde\psi_{i}(t)\right) \\
\Delta_{i}(t) &=
-4\theta_{i}(t)\tilde\psi_{i}(t)(h_{i}-\theta_{i}(t)
+a_{i,i}\tilde\psi_{i}(t))+(\theta_{i}(t)-a_{i,i}\tilde\psi_{i}(t)
+h_{i}\tilde\psi_{i}(t)-\theta_{i}(t)\tilde\psi_{i}(t))^{2},
\end{align*}
for $i=1,2$. We have the following result:
\begin{corollary}
Assume $N = 2$ and the rate matrix $A$ to be homogenous. Then, the solution functions $K(t,1)$ and $K(t,2)$ of (\ref{eq:postdeflog}) are given by
\begin{eqnarray}
\nonumber K(t,i) &=& \frac{1}{\aPlus^2} \bigg[ \zeta_{3-i}  a_{i,i} \left( 1 - e^{\aPlus (R-t)} + \aPlus  (R-t) \right)  \\
\nonumber & & + \zeta_i \left(a_{3-i,3-i}^2 (R-t) + a_{i,i}\left( e^{\aPlus (R-t) } -1 + a_{3-i,3-i} \left(R-t \right) \right) \right) \bigg].
\end{eqnarray}
The optimal feedback bond-fractions functions $\tilde{\pi}^P_{1}(t)$ and $\tilde{\pi}^P_{2}(t)$ are given by
\begin{equation}
{\tilde{\pi}^P_{i}(t)} =
\begin{cases}
-\frac{h_i + a_{ii} \tilde\psi_{i}(t)}{\left(h_i - a_{i,i} \right) \tilde\psi_{i}(t)}, & \text{if $\theta_i(t) = 0$,} \\
\frac{ \left(a_{ii} - h_i \right) \tilde\psi_{i}(t) + \theta_i(t) \left( \tilde\psi_{i}(t)-1 \right) + \sqrt{\Delta_{i}(t)}}{2 \theta_i(t) \tilde\psi_{i}(t)}, & \text{if $\theta_i(t) \neq 0$.}
\end{cases}
\end{equation}
Finally, the solution functions $J(t,1)$ and $J(t,2)$ of (\ref{eq:predeflog1}) are given by
\begin{align}
\nonumber J(t,i) &= \commterm \bigg[ \left( -\CPlusSQ \CTcosh + \mcthreeone \CTsinh \right) \\
\nonumber &\quad \times \IntTR e^{-\frac{  s \CPlusSQcot  }{2}} \bigg( 2 \Expcoct ( h_{i}-c_{i}) ( \zeta_{3-i}   + g_{3-i}(s)  ) \\
\nonumber   &+((3-2i)(c_{2}-c_{1})-\sqrt{c_{+}}-e^{s\sqrt{c_{+}}}((3-2i)(c_{2}-c_{1})+\sqrt{c_{+}})) (\zeta_i + g_{i}(s))  \bigg) ds\\
\nonumber &-2 (h_i - c_{i}) \CTsinh\times \IntTR e^{-\frac{  s \CPlusSQcot  }{2}} \bigg(\\
\nonumber &\quad (-\CPlusSQ + \ExpCPlusSQ [(3-2i)(c_{2}-c_{1})-\sqrt{c_{+}}]-(3-2i)(c_{2}-c_{1}) ) (\zeta_{3-i} + g_{3-i}(s)) \\
\nonumber &\quad +2 \Expcoct ( h_{3-i} - c_{3-i} ) ( \zeta_i   + g_{i}(s))   \bigg) ds
\bigg].
\end{align}

\end{corollary}
From the formulas given above, we can notice that for any given value of the historical intensity $h_1$, the optimal bond strategy $\tilde{\pi}^P_{1}(t)$ is independent of the value of the historical intensity $h_2$, associated with the other regime. A symmetric argument applies to $\tilde{\pi}^P_{2}(t)$, thus showing that each strategy depend on the other regime only through the risk-neutral loss and default intensity associated to the other regime.
This is not surprising, given that an investor would base his decision to buy or sell a defaultable bond on the market perception of default risk  (i.e. based on the risk-neutral default intensity parameters) rather than on the number of defaults experienced by the corporation in the past.

\subsection{{Economic Analysis}}
The objective of this section is to measure the impact of the default parameters over the value functions and the optimal bond strategy via numerical analysis. We fix the interest rate, drift, and volatility regime parameters as indicated in Table \ref{table:regimeparams}. We choose the transition rates of the chain to be the same under both probability measures. In order to evaluate $\psi_i(t)$, we use the analytical expression for the probability density $f_i(R,x)$ of the time spent in state $i$ by a two-regime continuous time Markov Chain, for a given time interval $[0,R]$, when the chain starts in state $i$ at time 0. Such formulas have been provided in \cite{Kovche}. Applying these formulas to our setting, for a given $0 \leq x \leq R$, we obtain
\begin{eqnarray*}
f_1(R,x) &=& e^{-a_{12}^{\Qx} R} \delta(x-R) + a_{12}^{\Qx} e^{-(a_{12}^{\Qx} x + a_{21}^{\Qx}(R-x))} \\
& & \left( I_0\left(2 \sqrt{\termbes}\right) + \sqrt{\frac{a_{12}^{\Qx} a_{21}^{\Qx} x}{R-x}} I_1\left(2 \sqrt{\termbes}\right) \right) \\
f_2(R,x) &=& e^{-a_{21}^{\Qx} R} \delta(x-R) + a_{21}^{\Qx} e^{-(a_{21}^{\Qx} x + a_{12}^{\Qx}(R-x))} \\
& & \left( I_0(2 \sqrt{\termbes}) + \sqrt{\frac{a_{21}^{\Qx} a_{12}^{\Qx} x}{R-x}} I_1(2 \sqrt{\termbes}) \right)
\end{eqnarray*}
where $\delta(x)$ denotes the Dirac delta function and $I_{j}(z)$ is a {modified Bessel function of first kind}.
When applying this formula to our case, we have
\begin{eqnarray*}
\psi_1(0) &=& \int_0^R e^{-(r_1 + h_1 L_1) x - (r_2 + h_2 L_2) (R-x) } f_1(R,x) dx \\
\psi_2(0) &=& \int_0^R e^{-(r_1 + h_1 L_1) (R-x) - (r_2 + h_2 L_2) x } f_2(R,x) dx
\end{eqnarray*}

We plot the behavior of $\tilde{\pi}^P_{1}(0)$ and $\tilde{\pi}^P_{2}(0)$ with respect to $h_1$ and $h_2$. It appears from Figure \ref{fig:optstratloghaz} that the strategies are not very sensitive to the starting regime when the ratio $\frac{h_1}{h_2}$ is close to one, with $h_1$ and $h_2$ not too large.
This is consistent with intuition, because in such a scenario, the default intensity is the same and small under both regimes, thus the default probability is also small and, consequently, the slightly different risk neutral loss rates in Table \ref{table:regimeparams} do not affect much the investment choice.
As the gap between $h_2$ and $h_1$ increases, the strategies $\tilde{\pi}^P_{1}(0)$ and $\tilde{\pi}^P_{2}(0)$ behave differently. This is because a larger risk-neutral intensity translates into a larger risk-neutral default probability, and given that the risk-neutral transition rates of the chain are not large, the starting regime matters.



\begin{table}
\centering
\begin{tabular}{|l|c|c|}
\hline
    & Regime `1' & Regime `2' \\
\hline
$ r $ &  0.03 & 0.03 \\
\hline
$\mu $ & 0.07 & 0.02  \\
\hline
$\sigma$ &  0.2 & 0.2 \\
\hline
\hline
\end{tabular}\caption{Parameters associated to the two regimes.}
\label{table:regimeparams}
\end{table}

\begin{figure}
\centering
\includegraphics[width=0.47\textwidth, viewport=100 200 500 500]{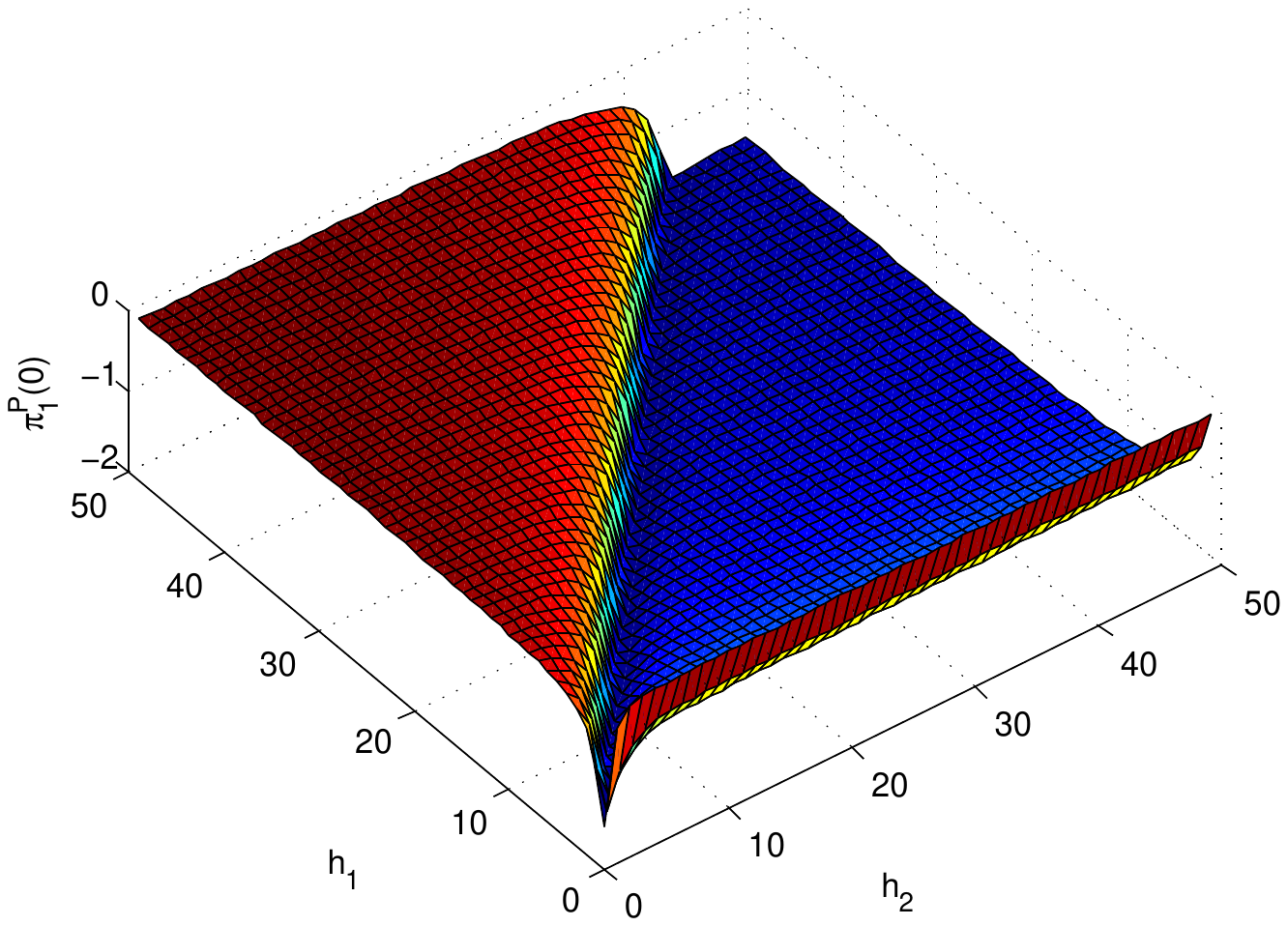}
\includegraphics[width=0.47\textwidth, viewport=100 200 500 500]{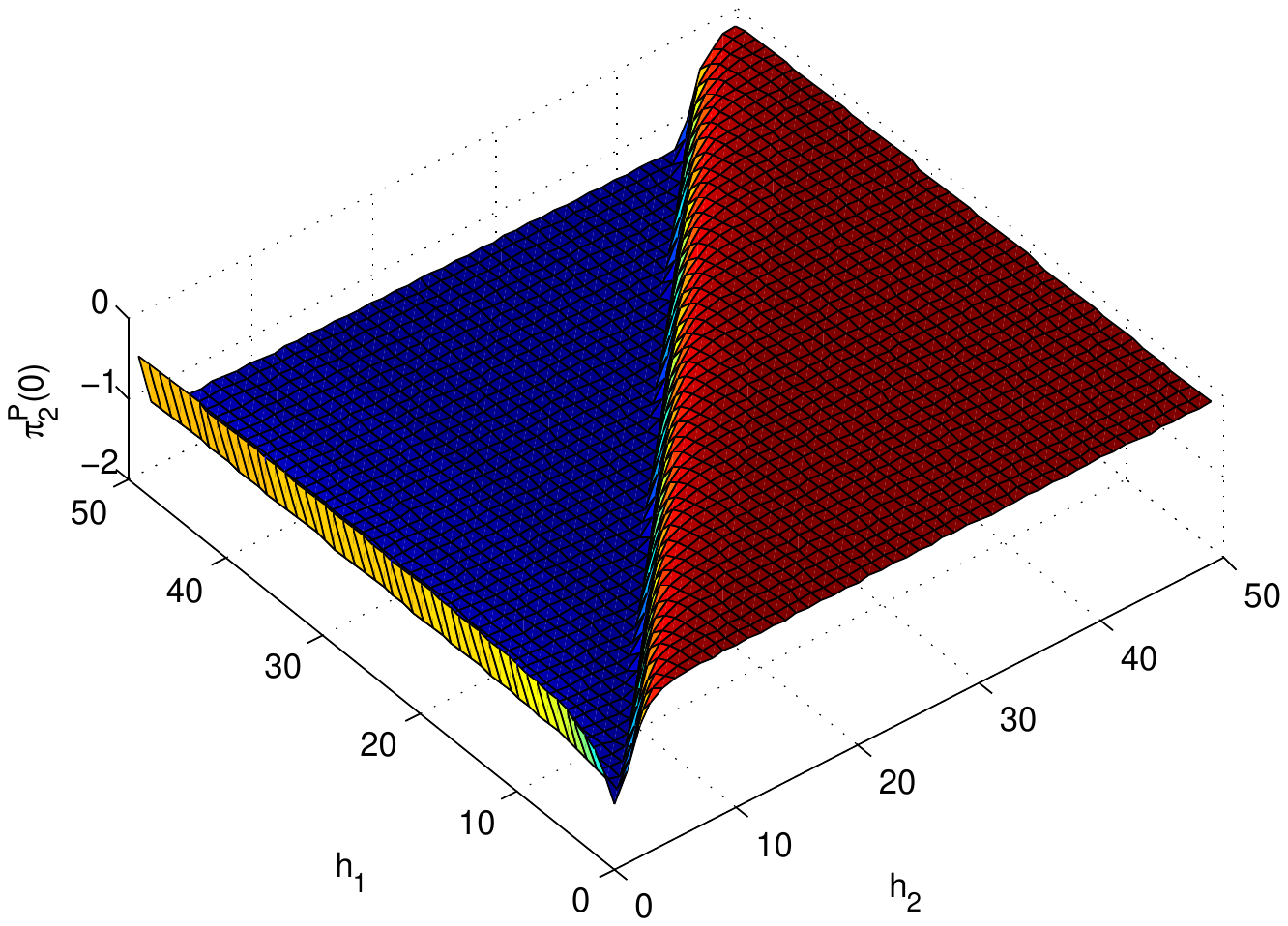}
\caption{Optimal bond strategy $\tilde{\pi}^P_1(0)$ and $\tilde{\pi}^P_2(0)$ versus the risk neutral hazard intensities $h_1$ and $h_2$. The upper surface represents $\tilde{\pi}^P_1(0)$, while the lower surface  represents $\tilde{\pi}^P_2(0)$. The loss parameters are given, respectively by $L_1 = 0.4$, and $L_2 = 0.45$. The transition rates are $a_{12} = 0.7$, and $a_{21} = 0.1$.
The horizon {$R$} is set to 2 years.}
\label{fig:optstratloghaz}
\end{figure}
We next show the behavior of the optimal bond strategy over time. It appears from the left panels of Fig.~\ref{fig:optstratTime} that the number of bond units sold decrease as the investment horizon increases. It can be noticed that, the riskier the corporate bond, the smaller the number of units sold. This is negatively correlated with the regime conditioned bond prices, as it can be checked from Fig.~\ref{fig:optstratTime}, showing that riskier bonds have smaller prices, and shorter maturity bonds have larger prices. Moreover, since the second regime is riskier and the transition rate from the second to the first regime very small, the price decay over time of the bond is faster when the Markov chain starts in the second regime ($\psi_2(t) < \psi_1(t)$), thus leading to a larger number of bond units sold when the chain starts in the second regime.

\begin{figure}
\includegraphics[width=0.5\textwidth, viewport=100 200 500 500]{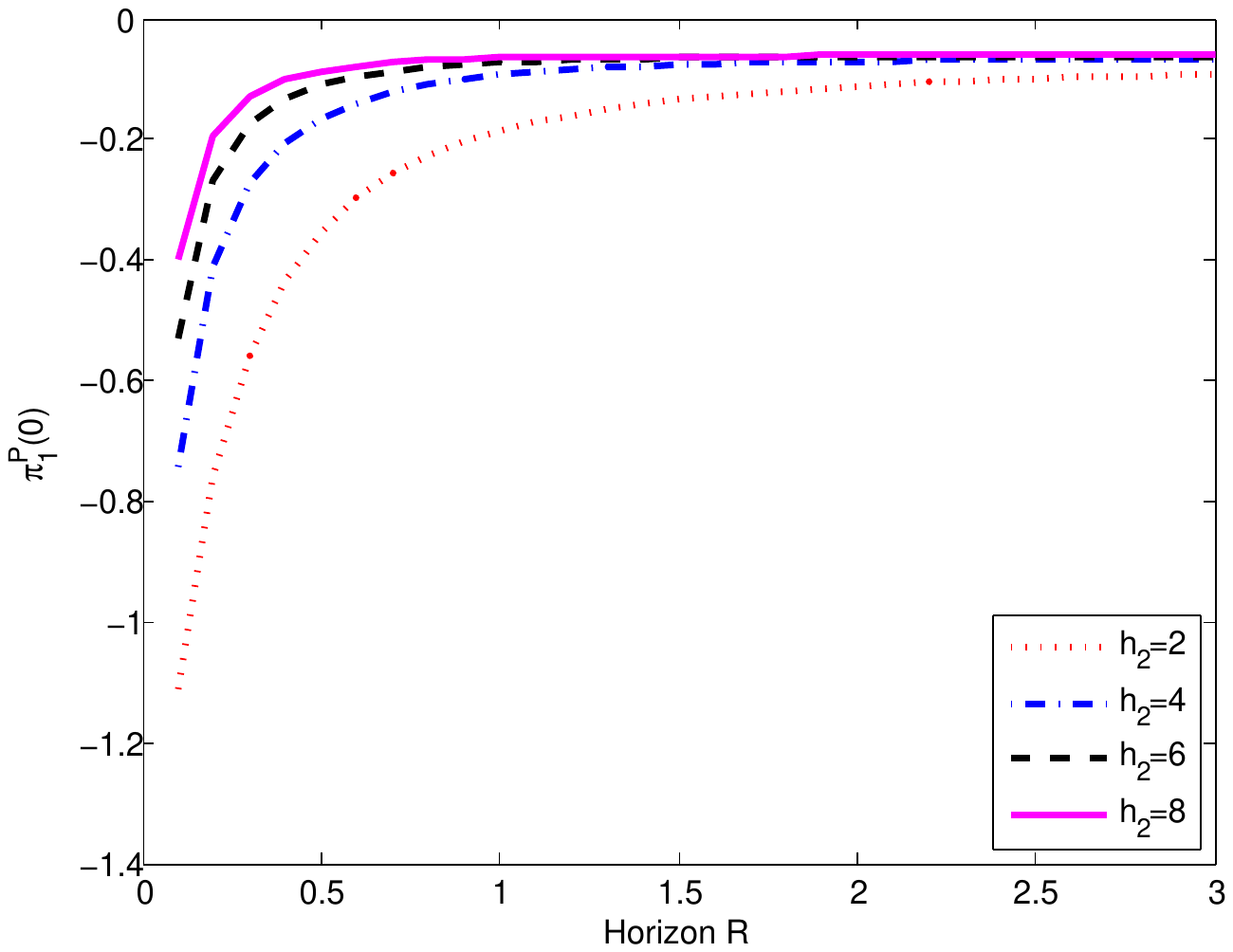}
\includegraphics[width=0.5\textwidth, viewport=100 200 500 500]{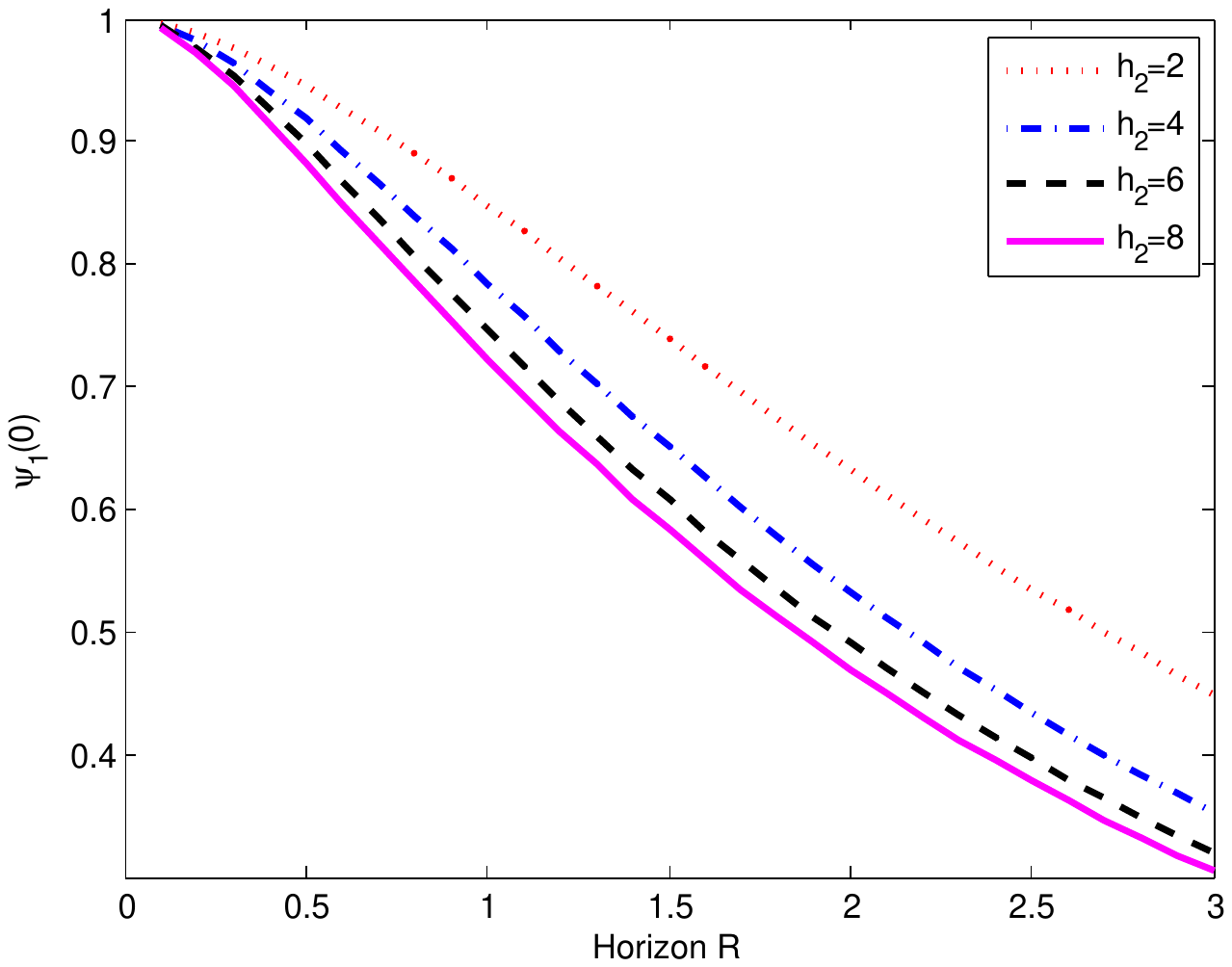}\\
\includegraphics[width=0.5\textwidth, viewport=100 200 500 500]{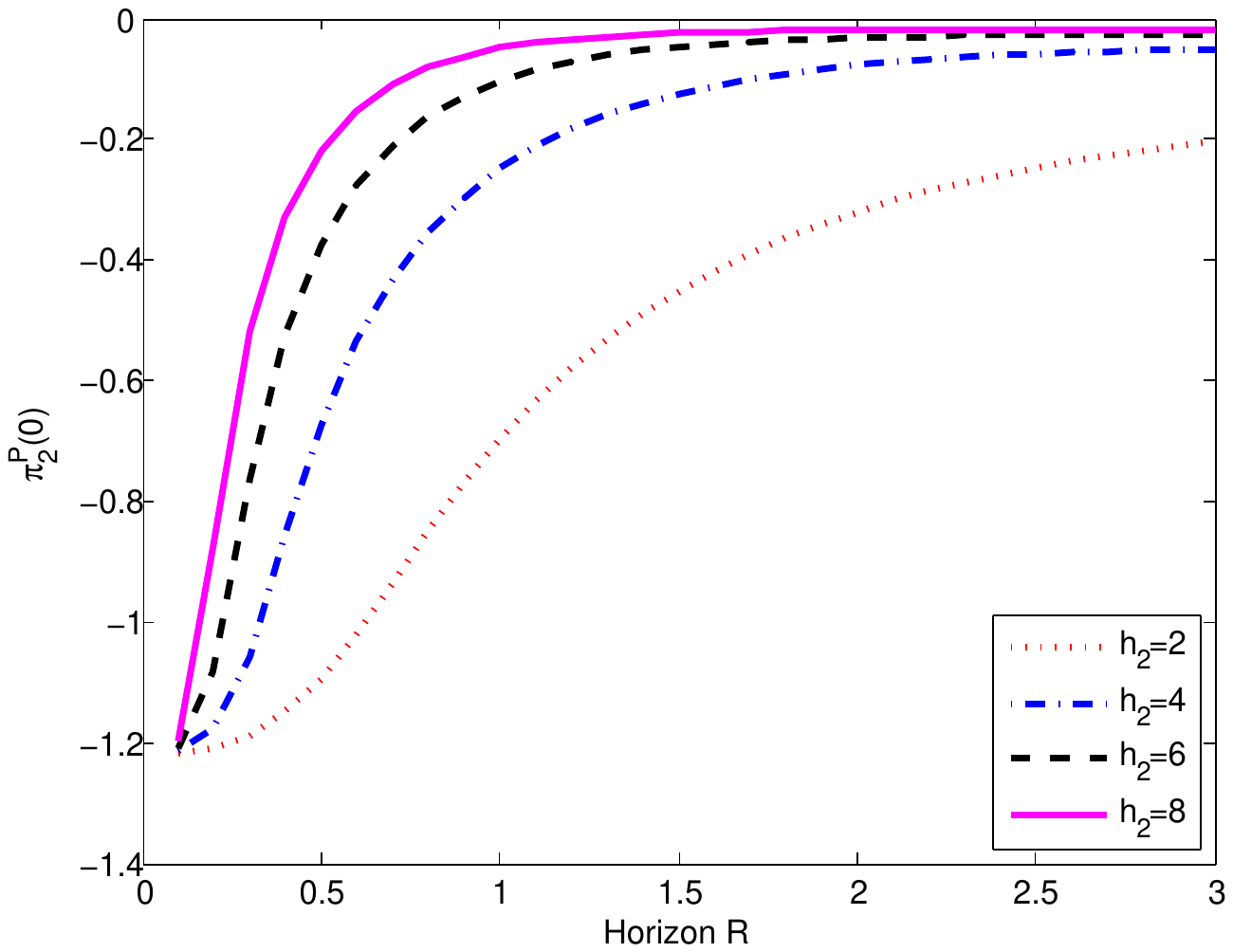}
\includegraphics[width=0.5\textwidth, viewport=100 200 500 500]{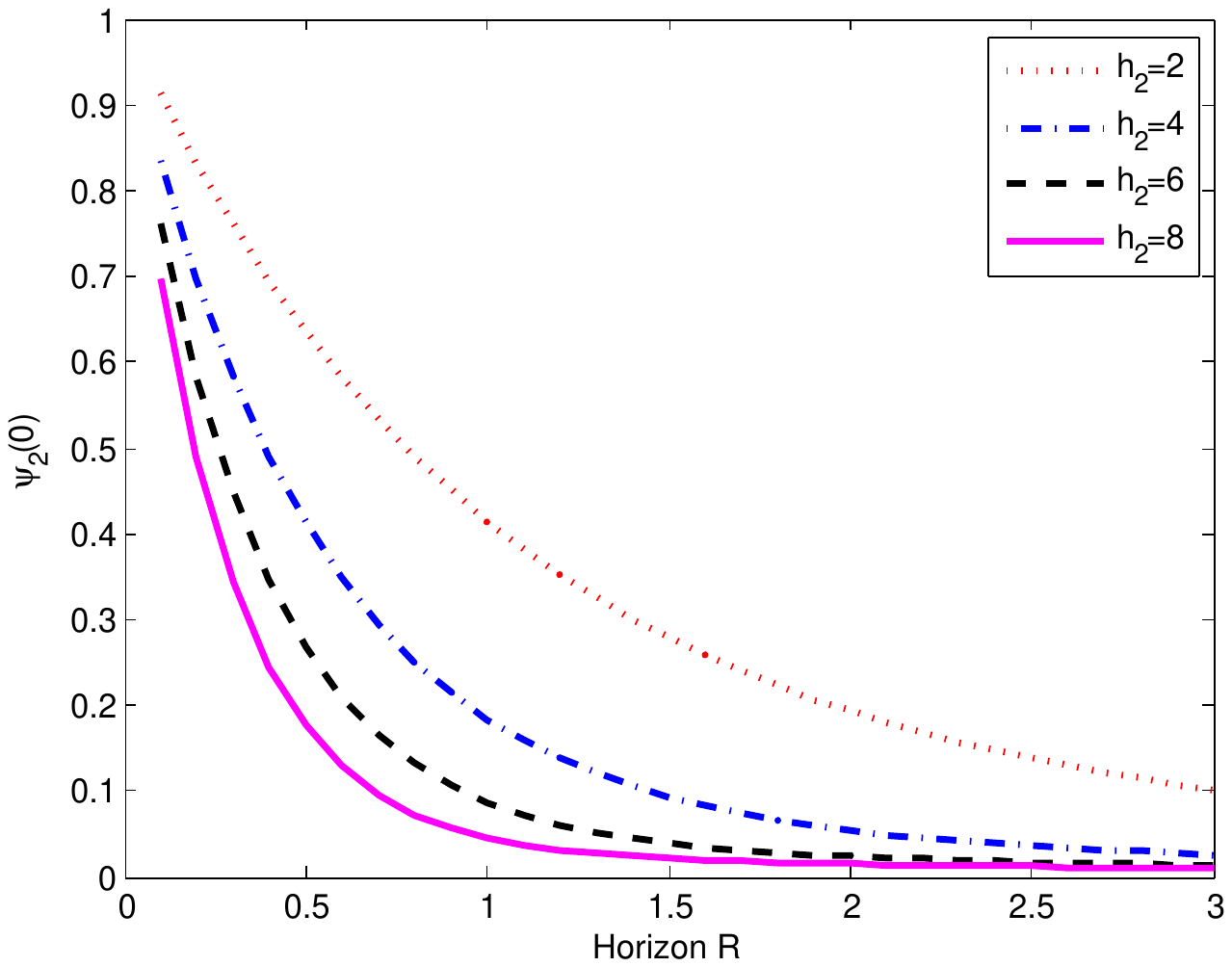}
\caption{Optimal bond strategies and regime-conditioned bond prices at time zero versus time horizon. The top left panel represents the optimal bond strategy $\tilde{\pi}^P_1(0)$. The top right panel represents the regime-conditioned bond price $\psi_1(0)$. The bottom left panel represents the optimal bond strategy $\tilde{\pi}^P_2(0)$. The bottom right panel represents the regime-conditioned bond price $\psi_2(0)$. The loss parameters are given, respectively, by $L_1 = 0.4$, and $L_2 = 0.45$. The hazard intensity $h_1 = 0.1$. The transition rates are $a_{12} = 0.4$, and $a_{21} = 0.1$.
}
\label{fig:optstratTime}
\end{figure}

We next inspect the behavior of the time components, $K(t,1)$ and $K(t,2)$, of the optimal post-default value function over time, for a fixed investment horizon $R$. We report the results in Figure \ref{fig:postdef}, and in each plot we superimpose the time component of the optimal value function obtained in the Merton model.
The latter is well known from the work of \cite{Merton}, and in our specific case obtained assuming that for each regime $i$, the market model consists of a stock and a money market account with parameters $\mu_i$, $r_i$, and $\sigma_i$. It appears from the plots that both post-default value functions decrease with time. Moreover, they differ at times $t$ far from the investment horizon due to possibility of regime shifts, and start converging to each other when the time to horizon is small. This is because the chain spends the largest fraction of its time in the starting regime due to the small transition rates ($a_{12}= 0.4$ and $a_{21}= 0.1$), and consequently for short times to the horizon, the wealth process of our regime switching model approaches the wealth process in the Merton model.

\begin{figure}
\includegraphics[width=0.5\textwidth, viewport=100 200 500 550]{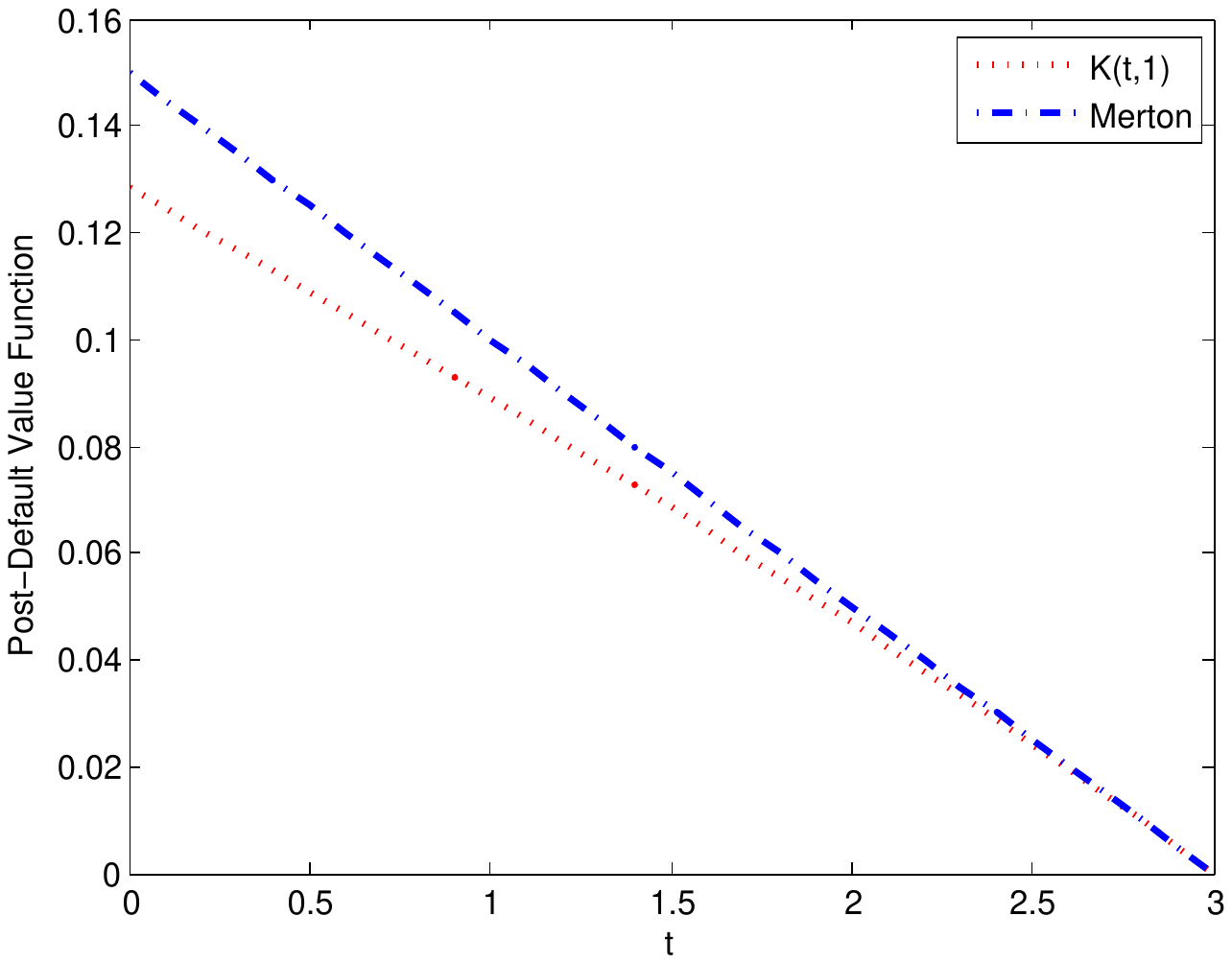}
\includegraphics[width=0.5\textwidth, viewport=100 200 500 550]{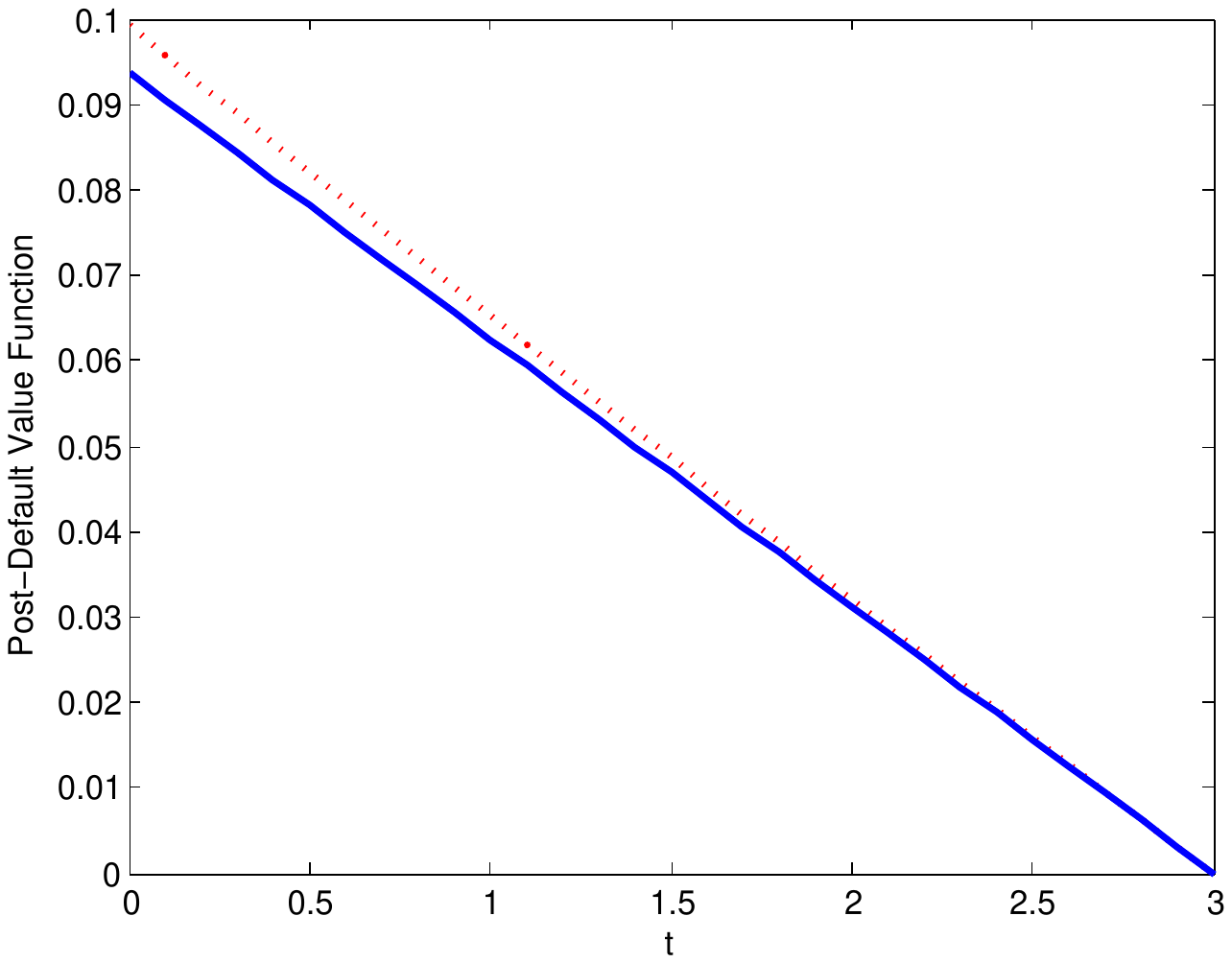}
\caption{The left panel represents the optimal post default value function $K(t,1)$, and the optimal value function in the Merton model with parameters $\mu_1$, $r_1$ and $\sigma_1$. The right panel represents the optimal post-default value function $K(t,2)$, and the optimal value function in the Merton model with parameters $\mu_2$, $r_2$ and $\sigma_2$. The transition rates are $a_{12} = 0.4$, and $a_{21} = 0.1$. The horizon $R$ is set to three years.
}
\label{fig:postdef}
\end{figure}

We finally evaluate the behavior of the time components, $J(t,1)$ and $J(t,2)$, of the optimal pre-default value function over time, assuming again a fixed investment horizon $R$. We report the results in Figure \ref{fig:predef}, where we vary $h_2$, keeping $h_1$ fixed. We notice from the plots that the pre-default value function is decreasing with time, and very sensitive to the default risk level. 

\begin{figure}
\includegraphics[width=0.5\textwidth, viewport=100 200 500 500]{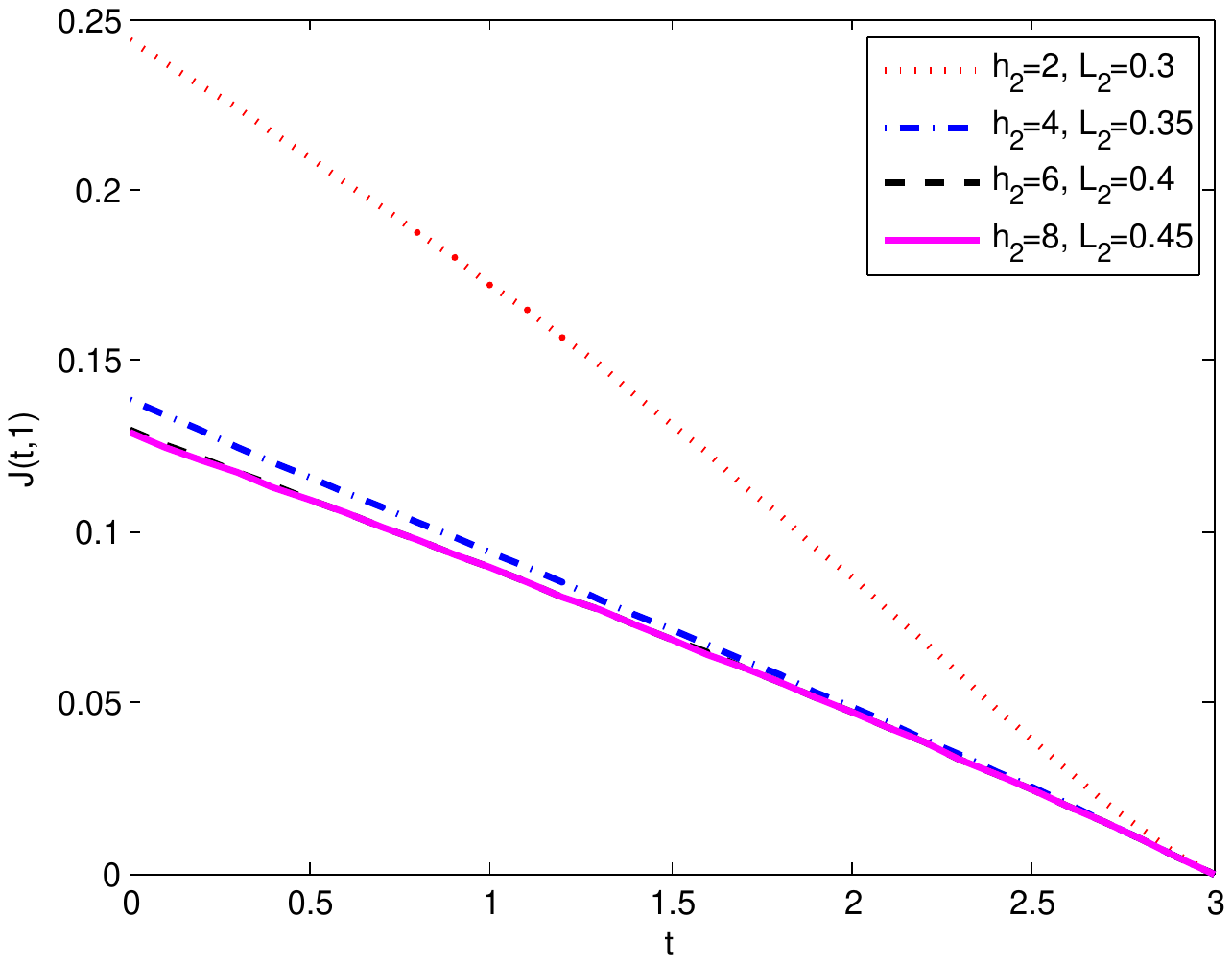}
\includegraphics[width=0.5\textwidth, viewport=100 200 500 500]{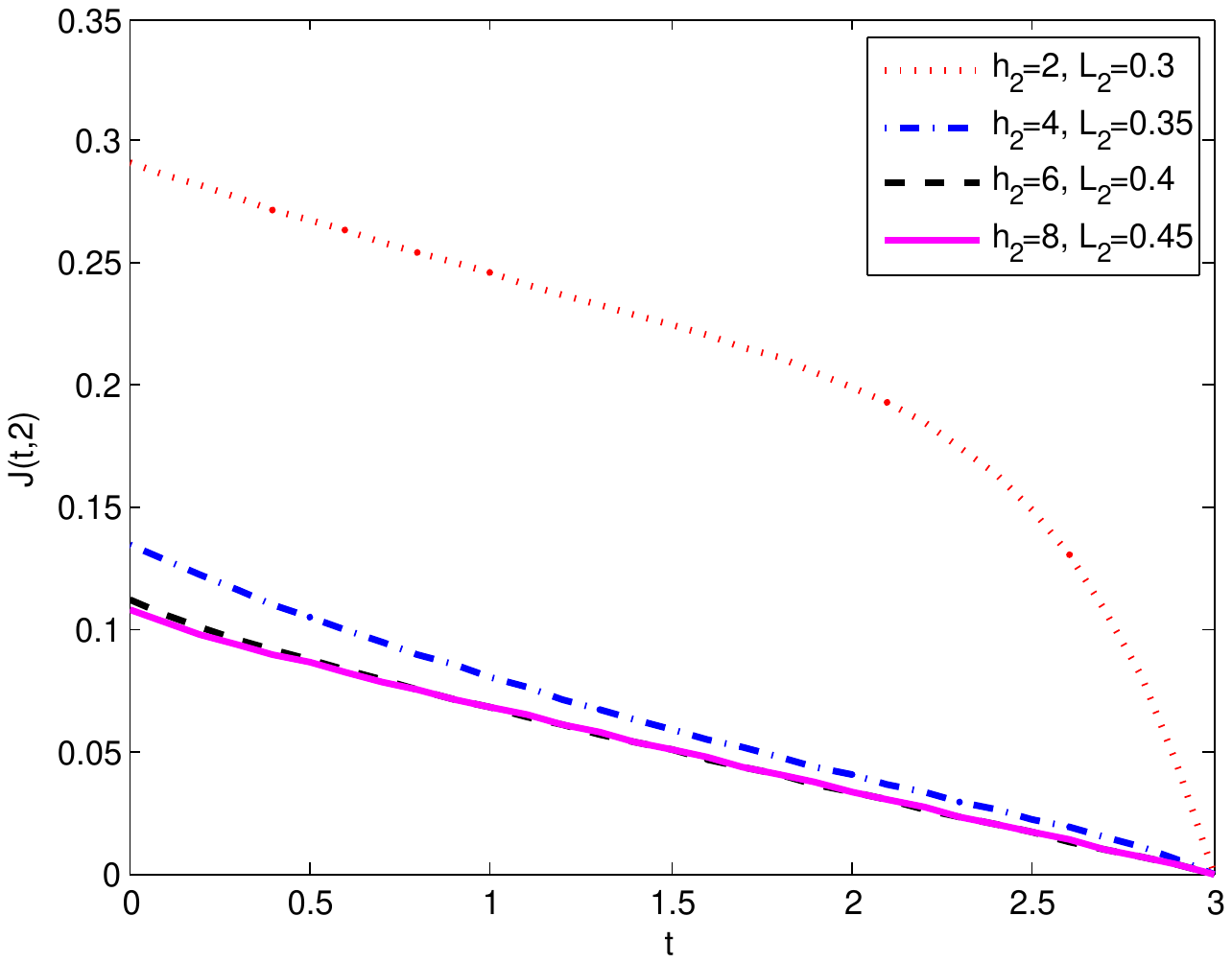}
\caption{The left panel represents the optimal pre-default value function $J(t,1)$ for different levels of default risk associated to regime 2, and parameterized by the pair $(h_2, L_2)$. The right panel represents the optimal pre-default value function $J(t,2)$ for the same levels of default risk $(h_2, L_2)$ as in the left panel. The transition rates are $a_{12} = 0.4$, and $a_{21} = 0.1$. The horizon $R$ is set to three years.}
\label{fig:predef}
\end{figure}

\section{Conclusions}\label{sec:conclusion}
We considered the continuous time portfolio optimization problem in a defaultable market, consisting of a stock, defaultable bond, and money market account. We assumed that the price dynamics of the assets are governed by a regime switching model.
We derived the dynamics of the defaultable bond under the historical measure from the risk neutral price process. We have shown that the utility maximization problem may be separated into a pre-default and a post-default optimization subproblem, and proven verification theorems for both cases under the assumption that the solutions are monotonic and concave in the wealth variable $v$. The post-default verification theorem shows that the optimal value function is the solution of a nonlinear Dirichlet problem with terminal condition. The pre-default verification theorem shows that the optimal pre-default value function and the optimal bond investment strategy
 can be obtained as the solution of a coupled system of nonlinear partial differential equations with terminal condition (satisfied by the pre-default value function) and nonlinear equations (satisfied by the bond investment strategy).
 {Each equation is associated to a different regime, and the dependence of a regime $i$ from another regime $j$ comes through the Markov transition rates and the ratio between the defaultable bond prices in regime $j$ and regime $i$. Our results imply that the pre-default optimal value function and the bond investment strategy depend on the optimal post-default value function.}

We demonstrated our framework on the concrete case of an investor with logarithmic utility, and shown that {both the optimal pre-default and post-default value function can be obtained as the solution of a linear system of first order ordinary differential equations, while the optimal bond strategy can be uniquely recovered as the solution of a decoupled system of nonlinear equations, one for each regime}. We have also performed an economic analysis on a two-regime market model with homogenous transition rates, and investigated the impact of default risk on the optimal strategy and value functions. Our analysis has shown that the optimal number of bond units sold in each regime decreases with the riskiness of the bond perceived by the market, and that the number of bond units sold is smaller for larger investment horizons.
Although we have specialized our framework to the specific case of logarithmic utility, it is flexible enough to accommodate any concave increasing utility function. Therefore, the results derived in the verification theorems can be used to derive explicit or numerical solutions for pre-default and post-default value functions, as well as optimal investment strategies, corresponding to a wide range of utilities of practical interest, such as the ones in the HARA family. 


\appendix

\section{Risk-neutral and historical bond dynamics}\label{app:dynamicsbond}
The following result states that the process {$\xi^{\Px}_t$} introduced in (\ref{MrtRprDftPrcP}) is also a {$\Qx$-martingale}. However, in order to indicate in the sequel when certain dynamics are being taken under $\Qx$ or under  $\Px$, we will introduce a new notation {$\xi^{\Qx}_t$}. We should keeping in mind through the proof below that $\xi^\Qx=\xi^\Px$.
\begin{lemma}\label{LmNdFD}
	The process
	\begin{equation}\label{MrtRprDftPrcQ}
		{\xi^{\Qx}_t} :=H(t) - \int_0^t (1-H(u^{-})) {h_u} du
	\end{equation}
	is also a $({\Gx},\Qx)-$(local) martingale.
\end{lemma}
\begin{proof}
By the definition of $\Qx$ and the fact that $\eta$ is a ${\Gx}$-martingale (see the {paragraph} before (\ref{DfnDnsty})),  it suffices to prove that $\eta_{t} \xi_{t}^{\Qx}$ is a ${\Gx}$-martingale under $\Px$. From It\^o's formula and the definition of $\eta$ in (\ref{DfnDnsty}), we have the process
	\begin{align*}
		\eta_{t} {\xi^{\Qx}_t}={\xi^{\Qx}_0}+\int_{0}^{t} \eta_{s^{-}} {d\xi^{\Qx}_s} +
		\int_{0}^{t} \xi_{s^{-}}^{\Qx} d\eta_{s}
		+\sum_{s\leq{}t} \Delta \xi_{s}^{\Qx} \Delta \eta_{s}.
	\end{align*}
	From (\ref{MrtRprDftPrcP}), (\ref{DfnDnsty}), and (\ref{MrtRprDftPrcQ}), {$(\eta_{t} {\xi^{\Qx}_t})_{t}$ can be written {as}}
	\begin{align*}
		{\xi^{\Qx}_0} +\int_{0}^{t} \eta_{s^{-}} d{\xi^{\Px}_s} +
		\int_{0}^{t} \xi_{s^{-}}^{\Qx} d\eta_{s}+\sum_{0<s\leq{}t} \sum_{k,l=1}^{N} \eta_{s^{-}} \kappa_{k,l}(u) \Delta H(s)
		\Delta H_{s}^{k,l},
	\end{align*}
	where $H_{t}^{k,l}:=\sum_{0<s\leq{}t}\idc_{\{X_{s^{-}}=k\}} \idc_{\{X_{s}=l\}}$. Since the first two terms on the right-hand side of the previous equality are (local) martingales under $\Px$, it remains to show that the last term vanishes. But, given that $\Delta H_{s}\neq{}0$ at $s=\tau$, the summation in the last term above will be $0$ provided that
	\(
		{\Delta X_{\tau}=0,}
	\)
	a.s.
	In order to show this, let us recall that by definition $X$ has no fixed-jump times; i.e. $\Px(\Delta X_{t}\neq{}0)=0$ for any fixed time $t>0$. Also, using the definition of $\tau$ given in Eq.~(\ref{eq:taudef}), $\tau=\inf\{t\geq{}0: \int_{0}^{t} h(X_{s})ds\geq{}\xi\}$, where $\xi$ is an exponential random variable independent of $X$. Then, conditioning on $X$,
	\(
		{\Px(\Delta X_{\tau}\neq{}0)=\Ex\left\{\Ex\left[\left.{\bf 1}_{\{\Delta X_{\tau}\neq{}0\}}\right|X_{s},s\geq{}0\right]\right\}}.
	\)
	Denoting $0<\tau_{1}<\tau_{2}<\dots$ the transition times of the Markov chain $X$, {$\Ex\left[\left.{\bf 1}_{\{\Delta X_{\tau}\neq{}0\}}\right|X_{s},s\geq{}0\right]$ is given by}
	\begin{align*}
		\Ex\left[\left.\sum_{i=1}^{\infty}{\bf 1}_{\{\tau=\tau_{i}\}}\right|X_{s},s\geq{}0\right]= \Ex\left[\left.\sum_{i=1}^{\infty}{\bf 1}_{\{\int_{0}^{\tau_{i}} h(X_{s})ds=\xi\}}\right|X_{s},s\geq{}0\right]=0,
	\end{align*}
	where the last equality {follows} from the independence of $X$ and $\xi$, and the fact that $\xi$ is a continuous random variable.
\end{proof}

\begin{lemma}\label{Lm:DfPsi}
	 Under the assumptions of Theorem \ref{Prop:DynBndQ},
	the function
	\begin{equation}\label{OrgDfnPsi2}
		\psi_i(t) = \Ex^{\Qx} \left[\left.e^{-\int_t^{T} (r_s + 		h_s L_s) ds} \right| {X_t = e_i} \right]
	\end{equation}
	is differentiable for any $t\in(0,T)$.
\end{lemma}
\begin{proof}
	Define functions $\tilde{k}_{i,j}:[0,\infty)\to(-1,\infty)$ such that
	\[
		\frac{1}{N-1}=a_{i,j}^{\Qx}(t) (1+\tilde\kappa_{i,j}(t)),\quad\text{ for } i\neq{}j, \quad\text{and}\quad
		\tilde{\kappa}_{i,i}=0.
	\]
	We also let $\tilde{a}_{i,j}:= 1/(N-1)$ for $i\neq{}j$ and $\tilde{a}_{i,i}= -1$, so that $\widetilde{A}:=\left[\tilde{a}_{i,j}\right]_{i,j=1,\dots,N}$ is a valid generator of a homogeneous Markov process with transition times determined by a homogeneous Poisson process and an embedded Markov chain $\{\widetilde{X}_{i}\}_{i\geq{}1}$ with transition probabilities $p_{i,j}:=1/(N-1)$ for $i\neq{}j$.  Now, let us define a probability measure $\widetilde{\Qx}$ with Radon-Nikod\'yn density $(\tilde\eta_{t})_{t}$ given by
\begin{equation}\label{DfnDnstyHom}
	\tilde\eta_{t}=1+\int_{(0,t]} \sum_{i,j=1}^{N} \tilde\eta_{u^{-}} \tilde\kappa_{i,j}(u) d \widetilde{M}_{u}^{i,j},
\end{equation}
where
\(
	\widetilde{M}_{t}^{i,j}:= H_{t}^{i,j}-\int_{0}^{t} a_{i,j}^{\Qx}(u)H^{i}_{u}du,
\)
and we used notation (\ref{JmpTrnPrc1}). By virtue of Proposition 11.2.3 in \cite{bielecki01},  $\{X_{t}\}_{t\geq{}0}$ is a continuous Markov process with generator $\widetilde{A}$ under $\widetilde{Q}$.
Next, note that $\psi(t):=(\psi_{1}(t),\dots,\psi_{N}(t))$ is such that
	\begin{equation}
		v(t,T) := \Ex^{\mathbb{Q}}\left[e^{-\int_t^T (r_s + h_s L_s) ds} | X_t \right]=\left<\psi(t),X_{t}\right>.
	\end{equation}
But also, changing into the probability measure $\widetilde{\Qx}$, we can write
\[
	v(t,T) = \Ex^{{\Qx}}\left[\left. e^{-\int_t^T (r_s + h_s L_s) ds} \right| \mathcal{F}_t \right]
	=\Ex^{\widetilde{\Qx}}\left[\left. e^{-\int_t^T (r_s + h_s L_s) ds}  \frac{\tilde{\eta}_{_{t}}}{\tilde\eta_{T}}\right| \mathcal{F}_t \right],
\]	
and, hence, we have the following representation for $\psi_{i}$:
\[
	\psi_{i}(t)  =\Ex^{\widetilde{\Qx}}\left[\left. e^{-\int_t^T (r_s + h_s L_s) ds}  \frac{\tilde{\eta}_{_{t}}}{\tilde\eta_{T}}\right| X_{t}=e_{i}\right].
\]
Recall that the solution of (\ref{DfnDnstyHom}) can be written as
\[
	{\tilde{\eta}_{t}:=e^{-\int_{0}^{t} \sum_{i,j}a_{i,j}^{\Qx}(u)\tilde{\kappa}_{i,j}(u)H^{i}_{u}du
	+\sum_{0<u\leq{}t} \log\left(1+\sum_{i,j}\tilde\kappa_{i,j}(u)\Delta H^{i,j}_{u}\right)}}.
\]
Let $\widetilde{K}(t)=[\widetilde{K}_{i,j}(t)]_{i,j}$ and $\tilde{r}(t):= (\tilde{r}_{1}(t),\dots,\tilde{r}_{N}(t))'$ be defined by
\[
	\widetilde{K}_{i,j}(t):=\log\left(1+\tilde\kappa_{i,j}(t)\right),\quad
	 \tilde{r}_{i}(t):=r_{i}+h_{i}L_{i} - \sum_{j=1}^{N}a_{i,j}^{\Qx}(t)\tilde\kappa_{i,j}(t).
\]
Then, we have
\[
	\psi_{i}(t)=\Ex^{\widetilde{\Qx}}\left[\left. \exp\left\{ -\int_t^T \tilde{r}(s)' X_{s} ds
	-\sum_{s\in(t,T]:\Delta X_{s}\neq{}0}  X_{s^{-}}'\widetilde{K} (s) X_{s}\right\}\right| X_{t}=e_{i}\right],
\]
where $X'$ denotes the transpose of $X$. Next, using that $(X_{t})_{t}$ is a homogeneous Markov process under $\widetilde{\Qx}$,
\[
	\psi_{i}(t)=\Ex^{\widetilde{\Qx}}_{i}\left[ \exp\left\{ -\int_0^{T-t} \tilde{r}(t+s) X_{s} ds
	-\sum_{s\in(0,T-t]:\Delta X_{s}\neq{}0}  X'_{s^{-}}\widetilde{K} (t+s) X_{s}\right\}\right],
\]
where we used the notation $\Ex^{\widetilde{\Qx}}_{i}\left(\cdot\right):=\Ex^{\widetilde{\Qx}}\left(\left.\cdot\right| X_{0}=e_{i}\right)$.
Furthermore, in terms of the transition times $\zeta_{1}<\zeta_{2}<\dots$ of $X$,  the embedded Markov chain $\{\widetilde{X}_{i}\}$ of $X$, and the number $M_{t}$ of transitions by time $t$  of $X$, $\psi_{i}(t)$ can be written as
\[
	\psi_{i}(t)=\left.\Ex^{\widetilde{\Qx}}_{i}\left[ e^{
	-\sum_{n=0}^{M_{\zeta}} \int_{\zeta\wedge\zeta_{n}}^{\zeta\wedge\zeta_{n+1}} \tilde{r}(t+s)' \widetilde{X}_{n} ds
	-\sum_{n=1}^{M_{\zeta}}  \widetilde{X}'_{n-1}\widetilde{K} (t+\zeta_{n}) \widetilde{X}_{n}}\right]\right|_{\zeta=T-t},
\]
where $\zeta_{0}=0$.
Using that $(M_{t})_{t}$ is a Poisson process under $\widetilde{\Qx}$ and conditioning on $M_{\zeta}$, we have
\begin{align*}
	{\psi_{i}(t)=\sum_{m=0}^{\infty} e^{-\zeta}\frac{\zeta^{m}}{m!}
	\Ex^{\widetilde{\Qx}}_{i}\left[ e^{
	-\sum_{n=0}^{m} \int_{\zeta U_{(n)}}^{\zeta U_{(n+1)}}\tilde{r}(t+s)'\widetilde{X}_{n} ds
	-\sum_{n=1}^{m} \widetilde{X}'_{n-1}\widetilde{K} (t+\zeta U_{(n)}) \widetilde{X}_{n}}\right]},
\end{align*}
where $\zeta=T-t$, $U_{(1)}<U_{(2)}<\dots<U_{(m)}$ are the ordered statistics of $m$ i.i.d. uniform $[0,1]$ variables independent of $\widetilde{X}$, $U_{(0)}=0$, and $U_{(m+1)}=1$. From the previous expression, we see that it suffices to show that
\begin{equation}\label{AuxDfn1}
	\Phi_{m}(\zeta):=\Ex^{\widetilde{\Qx}}_{i}\left[ e^{
	-\sum_{n=0}^{m} \int_{\zeta U_{(n)}}^{\zeta U_{(n+1)}}\tilde{r}(T-\zeta+s)' \widetilde{X}_{n} ds
	-\sum_{n=1}^{m} \widetilde{X}'_{n-1}\widetilde{K} (T-\zeta+\zeta U_{(n)}) \widetilde{X}_{n}}\right],
\end{equation}
is continuously differentiable in $\zeta\in(0,T)$ for each $m\geq{}0$, and that there exists a sequence $\{K_{m}\}_{m\geq{}0}$ such that
\[
	{\rm (i)}\;\sup_{0<\zeta< T}|\Phi_{m}(\zeta)|\leq{} K_{m},\quad
	{\rm (ii)}\;\sup_{0<\zeta< T}|\Phi_{m}(\zeta)|\leq{} K_{m},\quad
	{\rm (iii)}\;\sum_{m=0}^{\infty} e^{-\zeta}\frac{\zeta^{m}}{m!} K_{m}<\infty.
\]
We now show that (i)-(iii) are satisfied provided that
\[
	\sup_{s\in[0,T]} |\tilde{r}_{i}(s)|<\infty,\quad
	\sup_{s\in[0,T]} |\tilde{K}_{i,j}(s)|<\infty,\quad\text{ and }\quad
	\sup_{s\in[0,T]} |\tilde{K}'_{i,j}(s)|<\infty,
\]
are satisfied. The latter conditions directly follow from (\ref{NdCndDf1}).
Denoting $\Gamma_{m}(\zeta)$ the random function inside the expectation $\Ex^{\widetilde{\Qx}}_{i}$ in (\ref{AuxDfn1}), one can check that
\begin{equation}\label{Dmmy1}
	\sup_{0\leq\zeta\leq{}T}|\Gamma_{m}(\zeta)|\leq{}e^{m T\max_{i}\sup_{0\leq{}\zeta\leq T} |\tilde{r}_{i}(\zeta)'|
	+m \max_{i,j}\sup_{0\leq{}\zeta\leq T} |\widetilde{K}_{i,j}(\zeta)|}:=A^{m},
\end{equation}
for a constant $A<\infty$.
Also, $\Gamma_{m}(\zeta)$ is continuously differentiable and
\begin{align*}
	 \Gamma_{m}'(\zeta)&=\Gamma_{m}(\zeta)\left\{\sum_{n=0}^{m}\left[\bar{U}_{(n+1)} \tilde{r}(T-\zeta\bar{U}_{(n+1)})'\widetilde{X}_{n}
	 -\bar{U}_{(n)} \tilde{r}(T-\zeta\bar{U}_{(n)})' \widetilde{X}_{n}\right]\right.\\
	&\qquad\quad\qquad\left. - \sum_{n=1}^{m}\widetilde{X}'_{n-1}\widetilde{K}' (T-\zeta\bar{U}_{(n)}) \widetilde{X}_{n}\bar{U}_{(n)} \right\},
\end{align*}
where $\bar{U}_{(n)}:=1-U_{(n)}$. In particular, there exists a constant $B<\infty$ such that
\begin{equation}\label{Dmmy2}
	\sup_{0\leq\zeta\leq{}T}|\Gamma'_{m}(\zeta)|\leq{}  B m A^{m}.
\end{equation}
By the formal definition of the derivative $\Phi_{m}'(\zeta)$ and the dominated convergence theorem,  one can check that (\ref{Dmmy1}-\ref{Dmmy2}) will suffice for (i)-(iii).
\end{proof}

We are ready to give the proofs of the bond price dynamics:
\begin{proof}[Proof of Theorem \ref{Prop:DynBndQ}]
We first write the pre-default dynamics of the bond price under the risk neutral measure (i.e. on the event $\tau>t$). From Eq.~(\ref{eq:bondpricern}), we have
\begin{eqnarray*}
p(t,T) = \idc_{\tau > t} \Ex^{\mathbb{Q}}\left[e^{-\int_t^T (r_s + h_s L_s) ds} | \mathcal{F}_t \right]
= \idc_{\tau > t} \Ex^{\mathbb{Q}}\left[e^{-\int_t^T (r_s + h_s L_s) ds} | X_t \right].
\end{eqnarray*}
Define $v(t,T) = \Ex^{\mathbb{Q}}\left[e^{-\int_t^T (r_s + h_s L_s) ds} | X_t \right]$ and $\tilde{H}(t) = \idc_{\tau > t}$ so that
\(
	p(t,T) = \tilde{H}_t v(t,T).
\)
In terms of (\ref{OrgDfnPsi}), note that $v(t,T)= \left<\psi(t), X_t\right>$, where {$\psi(t)$ is given as in (\ref{OrgDfnPsi})}. In particular,
\(
	\Delta v(t,T)=\left<\psi(t), \Delta X_t\right>,
\)
since $\psi(t)$ is continuous in light of Lemma \ref{Lm:DfPsi}.
Next, let us introduce the {$({\Fx},\Qx)$}-martingale
\(
	\phi(t) := \Ex^{\mathbb{Q}}\left[e^{-\int_0^T (r_s + h_s L_s) ds} |\mathcal{F}_t \right],
\)
and the process $\tilde{b}(t) := \exp\{\int_0^t (r_s + {h}_s L_s) ds\}$. Then, $v(t,T)=\tilde{b}(t) \phi(t)$ has dynamics
\(
dv(t,T) = (r_t + {h}_t L_t) v(t,T) dt + \tilde{b}(t) d\phi(t).
\)
This leads to
\begin{eqnarray}\label{NdEq1}
dp(t,T) &=& \tilde{H}(t^{-}) dv(t,T) + {v(t^{-},T)} d\tilde{H}(t) + \Delta v(t,T) \Delta \tilde{H}(t).
\end{eqnarray}
By virtue of the identity $\Delta v(t,T)=\left<\psi(t), \Delta X_t\right>$ and similar arguments to those in the proof Lemma \ref{LmNdFD}, we have that
\[
	{\sum_{0<s\leq{}t}\Delta v(s,T)\Delta \tilde{H}(s)=\sum_{0<s\leq{}t}\left<\psi_{s},\Delta X_{s}\Delta \tilde{H}(s)\right>=\left<\psi_{\tau},\Delta X_{\tau}\right>=0}.
\]
Thus, (\ref{NdEq1}) simplifies as follows:
\begin{align}
\nonumber
dp(t,T)&={ \tilde{H}(t^{-}) \left[ (r_t + h_t L_t) v(t,T) dt + \tilde{b}(t) d\phi(t) \right] + v(t^{-},T) d\tilde{H}(t)}\\
&={(r_t + h_t L_t) p(t^{-},T) dt + \tilde{H}(t^{-}) \tilde{b}(t) d\phi(t) - v(t^{-},T) dH(t).}
\end{align}
Let us now try to find the dynamics of $\phi(t)$. Since $\phi_t = e^{-\int_0^t (r_s + h_s L_s) ds} \left<\psi(t), X_t\right>$, where $\psi(t) $ is given as in (\ref{OrgDfnPsi}), It\^o's formula leads to
\begin{align*}
d\phi(t) &= -(r_t + h_t L_t) e^{-\int_0^t (r_u + h_u L_u) du} \left<{\psi(t)}, X_t\right> dt  \\
&\quad + e^{-\int_0^t (r_u + h_u L_u) du} \left(\left< \frac{d\psi(t)}{dt}, X_t\right> + \left<\psi(t), (A^{\Qx})'_{t} X_t\right>\right) dt\\
&\quad + e^{-\int_0^t (r_u + h_u L_u) du} {\left<{\psi(t)}, dM^{\Qx}(t)\right>},
\end{align*}
where we had used the differentiability proved in Lemma \ref{Lm:DfPsi} and the semi-martingale representation formula of our Markov chain $X_t$ given in Eq.~(\ref{eq:MCsemimb}).
As $\phi(t)$ is a {$({\Fx},\Qx)$-(local)} martingale, its drift term is zero and, therefore, we obtain the dynamics
\begin{equation}
d\phi(t) = e^{-\int_0^t (r_u + h_u L_u) du} {\left<{\psi(t)}, dM^{\Qx}(t)\right>}
\end{equation}
and all together, we get
\begin{equation}
dp(t,T) = (r_t + h_t L_t) p(t^{-},T) dt + \tilde{H}(t^{-}) {\left<{\psi(t)}, dM^{\Qx}(t)\right> - v(t^{-},T)} dH(t).
\end{equation}
By the Doob-Meyer decomposition of Lemma \ref{LmNdFD}, we have $dH(t) = d\xi_t^{\Qx} + (1-H(t^{-})) h_t dt$. The last step follows from the fact that on the event $\tau>t$, we have $v(t^{-},T)=p(t^{-},T)$, by definition of $p(t,T)$, and also $p(t^{-},T) =\left<\psi(t), X_{t^{-}}\right>$. Therefore, we can write the pre-default risk-neutral dynamics of the bond as
\begin{eqnarray}\label{DRNM2}
\frac{dp(t,T)}{p(t^{-},T)} &=& (r_t + h_t (L_t-1)) dt + \frac{{\left<{\psi(t)}, dM^{\Qx}(t)\right>}}{{\left<\psi(t), X_{t^{-}}\right>}} - d\xi_t^{\Qx}
\end{eqnarray}
\end{proof}

\begin{proof}[Proof of Proposition \ref{Prop:DynBndP}]
The bond price dynamics under the real-world measure follows directly from plugging Eq. (\ref{RMM0}) into (\ref{DRNM2}) to get:
\begin{align*}
\frac{dp(t,T)}{{p(t^{-},T)}} &= \left(r_t + h_t (L_t-1)+
\frac{\left<\psi(t),(A'(t)-(A^{\mathbb{Q}})'(t))X_{t}\right>}{
\left<\psi(t),X_{t}\right>} \right) dt\\
& \quad + \frac{\left<\psi(t), dM^{\Px}(t)\right>}{{\left<\psi(t), X_{t^{-}}\right>}}-d\xi^{\Px}_{t},
\end{align*}
where we also used that $\xi^{\Px}=\xi^{\Qx}$ (see Lemma \ref{LmNdFD}).
Finally,  we get the dynamics (\ref{DRWP0}) since
\begin{align}\label{AuxQnt3}
	 \frac{\left<\psi(t),(A'(t)-(A^{\mathbb{Q}})'(t))X_{t}\right>}{\left<\psi(t),X_{t}\right>}&=
	 \frac{\left<(A(t)-A^{\mathbb{Q}}(t))\psi(t),X_{t}\right>}{\left<\psi(t),X_{t}\right>}.
\end{align}
When $X_{t}=e_{i}$, the previous quantity equals
\[
	 D_{i}(t):=\sum_{j=1}^{n}(a_{i,j}(t)-a_{i,j}^{\Qx}(t))\frac{\psi_{j}(t)}{\psi_{i}(t)},
\]
and (\ref{AuxQnt3}) can be written as $\left<(D_{1}(t),\dots,D_{N}(t))',X_{t}\right>$.
\end{proof}

\section{{Derivation of the generator of $(t,V_{t},C_{t},H_{t})$}}\label{Sect:DrvGnr}
We start by changing our notation to be more consistent with the framework in \cite{bielecki01}. To this end, let
\begin{equation}\label{DfnChnPrc}
	{C_{t}:=\sum_{i=1}^{N} i {{\bf 1}_{\{X_{t}=e_{i}\}}}}.
\end{equation}
Note that $(C_{t})_{t}$ is a Markov process with values in $\{1,\dots, N\}$ and infinitesimal generator $A(t)=[a_{i,j}(t)]_{i,j=1,\dots,N}$. In particular, for any function ${g}:\{0,\dots,N\}\to\mathbb{R}$,
\[
	M^{{g}}(t):={g}(C_{t})-{g(C_{0})}-\int_{0}^{t} (A{g})(C_{u},u) du,
\]
is a martingale under $\Px$, where we have used the notation
\[
	{Ag}(i,t):= [{{\bf g}}' A'(t)]_{i} =\sum_{j=1}^{N} a_{i,j}(t) {g}(j), \quad \text{with}\quad {{\bf g}=(g(1),\dots,g(N))'};
\]
c.f. Proposition 11.2.2 in \cite{bielecki01}. Note that this result follows directly from the semimartingale decomposition (\ref{eq:MCsemim}) by multiplying {(from the left)} both sides  {there} by ${{\bf g}}'$. In particular, also note that
\begin{equation}\label{RBTN}
	M^{{g}}(t)={\bf {g}}' M^{\Px}(t), \quad \text{ and }\quad
	{M^{\Px}_{j}(t)=H_{t}^{j}-\int_{0}^{t} a_{_{C_{u},j}}(u) du},
\end{equation}
where we used notation (\ref{JmpTrnPrc1}).
Let us assume that $V$ admits the following Markov-modulated dynamics:
\begin{align}\label{GenDyn}
 dV_t= \alpha_{_{C_{t}}} dt+ \vartheta_{_{C_{t}}} d W_{t}
 +\sum_{j=1}^{N}\beta_{_{C_{t^{-}},j}} dM^{\Px}_{j}(t)
- \gamma_{_{C_{t^{-}}}} d\xi^{\Px}_{t},
\end{align}
where $\alpha_{i}(\cdot,\cdot,z),\vartheta_{i}(\cdot,\cdot,z),\beta_{i,j}(\cdot,\cdot,z),\gamma_{i}(\cdot,\cdot,z)$ are deterministic smooth functions in $[0,\infty)\times \mathbb{R} $ for any $i,j\in\{1,\dots,N\}$ and $z\in\{0,1\}$, and all the coefficients in  (\ref{GenDyn}) are evaluated at $(t,V_{t^{-}},H(t^{-}))$.  
 In terms of the processes (\ref{CrsMrt})-(\ref{JmpTrnPrc1}) and using (\ref{RBTN}), we first note that
\begin{align}
	&\sum_{j=1}^{N}\beta_{_{C_{t^{-}}},j} dM^{\Px}_{j}(t)=\sum_{i,j=1}^{N}
	\beta_{i,j}(t,V_{t^{-}},H(t^{-})) {\bf 1}_{\{C_{t^{-}}=i\}}  dM^{\Px}_{j}(t) \label{DcmpCngNotn}
	\\
	&\quad =\sum_{i,j=1}^{N}
	\beta_{i,j}(t,V_{t^{-}},H(t^{-})) {\bf 1}_{\{C_{t^{-}}=i\}} d H^{j}_{t}-\sum_{i,j=1}^{N}\beta_{i,j}(t,V_{t^{-}},H(t^{-})) H_{t}^{i}
	{a_{_{{C_{t}},j}}}(t) dt \nonumber
	\\
	&\quad =\sum_{i=1}^{N}\sum_{j\neq{}i}(\beta_{i,j}-\beta_{i,i})(t,V_{t^{-}},H(t^{-}))d H_{t}^{i,j}-\sum_{j=1}^{N}\beta_{_{C_{t},j}}(t,V_{t^{-}},H(t^{-}))
	{a_{_{C_{t},j}}}(t) dt.\nonumber
	\end{align}
and, in particular,
\begin{equation}\label{JmpWlthPrc}
	\Delta V_{t}=\sum_{i=1}^{N}\sum_{j\neq{}i}{\beta^{0}_{i,j}}(t,V_{t^{-}},H(t^{-}))\Delta H_{t}^{i,j} - \gamma_{C_{t^{-}}} (t,V_{t^{-}},H(t^{-}))\Delta H(t),
\end{equation}
where
\(
	{{\beta^{0}_{i,j}= \beta_{i,j}-\beta_{i,i}}}.
\)
Next, let {$f(\cdot,\cdot,i,z)\in C^{1,2}([0,\infty)\times \mathbb{R})$, for each} $i=1,\dots,N$ and $z\in\{0,1\}$. We want to find the semimartingale decomposition of $f(t,V_{t},C_{t},H(t))$. Applying the It\^o's formula (seeing $C_{t}$ as simply a bounded variation process), we have that
\begin{align}\label{SDl1}
	f(t,V_{t},C_{t},H(t))&=f(0,V_{0},C_{0},H(0))+ \int_{0}^{t}  f_{t}(u,V_{u},C_{u},H(u)) du\\
	&\quad+\int_{0}^{t}   \left\{\alpha_{_{C_{u}}}-{\sum_{j=1}^{N}\beta_{_{{C_{u}},j}}} a_{_{{C_{u}},j}}{+(1-H(u))  h_{_{C_{u}}}\gamma_{_{C_{u}}}}\right\} f_{v}du \nonumber\\
	&\quad+\int_{0}^{t}  \vartheta_{_{C_{u}}}  f_{v}dW_{u} + \frac{1}{2} \int_{0}^{t} f_{vv} \vartheta_{_{C_{u}}}^{2}du \nonumber\\
	&\quad+\sum_{0<u\leq{}t} \left\{f(u,V_{u},C_{u}, H(u))-f(u,V_{u^{-}},C_{u^{-}},H(u^{-}))\right\}.\nonumber
\end{align}
Since $\tau$ is not a transition time of $C$ a.s., we can write the last term in {the} above equation as follows:
\begin{align*}
	J_{t}&:=\sum_{0<u<{}t\wedge \tau} \left\{f(u,V_{u},C_{u}, 0)-f(u,V_{u^{-}},C_{u^{-}},0)\right\}\\
	&\quad+ \left\{f(\tau,V_{\tau^{-}}-\gamma_{_{C_{\tau^{-}}}}(\tau,V_{\tau^{-}},0),C_{\tau^{-}}, 1)-f(u,V_{\tau^{-}},C_{\tau^{-}},0)\right\} H(t)\\
	&\quad+\sum_{t\wedge \tau <u\leq{}t} \left\{ f(u,V_{u},C_{u}, 1)-f(u,V_{u-},C_{u^{-}},1)\right\}\\
	 &=\sum_{i=1}^{N}\sum_{j\neq{}i}\int_{0}^{t}\left[f(u,V_{u^{-}}+{\beta^{0}_{i,j}},j,H(u^{-}))-f(u,V_{u-},i,H(u^{-}))\right]
	d H_{u}^{i,j}\\
	&\quad+ \int_{0}^{t}   \left\{f(u,V_{u^{-}}-\gamma_{_{C_{u^{-}}}}(u,V_{u^{-}},0),C_{u^{-}}, 1)-f(u,V_{u^{-}},C_{u^{-}},0)\right\}
	d {H(u)}.
\end{align*}
Next, using the local martingales (\ref{MrtRprDftPrcP}) and (\ref{CrsMrt}), we have
\begin{align*}	
	J_{t}&= \sum_{i=1}^{N}\sum_{j\neq{}i}\int_{0}^{t}\left[f(u,V_{u^{-}}+{\beta_{i,j}^{0}},j,H(u^{-}))-f(u,V_{u-},i,H(u^{-}))\right]
	d M_{u}^{i,j}\\
	&\quad+ \int_{0}^{t}   \left\{f(u,V_{u^{-}}-\gamma_{_{C_{u^{-}}}}(u,V_{u^{-}},0),C_{u^{-}}, 1)-f(u,V_{u^{-}},C_{u^{-}},0)\right\}
	d \xi^{\Px}_{u}\\
	&\quad + \int_{0}^{t}
	 \sum_{j\neq{}C_{u}}{a_{_{C_{u},j}}}(u)\left[f(u,V_{u}+{\beta^{0}_{_{C_{u},j}}},j,H(u))-f(u,V_{u-},C_{u},H(u))\right]
	d u\\
	&\quad + \int_{0}^{t}   \left\{f(u,V_{u}-\gamma_{_{C_{u}}}(u,V_{u},0),C_{u^{-}}, 1)-f(u,V_{u},C_{u},0)\right\}
	(1-H(u)) {h_{_{C_{u}}}}du,
\end{align*}
where we had also used that $V_{u}=V_{u^{-}}$, $H(u)=H(u^{-})$, and $C_{u}=C_{u^{-}}$ a.e. and, hence, the integrands in the last two integrals with respect to $du$ can be evaluated at $(V_{u},C_{u},H(u)$ instead of $(V_{u^{-}},C_{u^{-}},H(u^{-})$.
All together, we have the semimartingale decomposition
\begin{align}\label{KSMD}
	{f(t,V_{t},C_{t},H(t))}&={f(0,V_{0},C_{0},H(0))}+ \int_{0}^{t} \mathcal{L}f(u,V_{u},C_{u},{H(u)}) du+
	{\mathcal{M}_{t}},
\end{align}
where {$(\mathcal{M}_{t})_{t}$ is a local martingale} and $ \mathcal{L} f(t,v,i,z) $ is the so-called generator of $(t,V_{t},C_{t},H(t))$ defined by
\begin{align}\nonumber
& \frac{\partial f}{\partial t} +
\frac{\partial f}{\partial v}\left\{\alpha_{i}(t,v,z)-{\sum_{j=1}^{N}\beta_{i,j}(t,v,z) {a_{i,j}}(t)}
{+(1-z) h_{i}\gamma_{i}}\right\}\\
&\quad+\frac{\vartheta^{2}_{i}(t,v,z)}{2}
 \frac{\partial^{2} f}{\partial v^2} \nonumber \\
 &\quad+ \sum_{j \neq i} {a_{i,j}}(t)
 \left( f(t,v+{\beta^{0}_{i,j}}(t,v,z),j,z) - f(t,v,i,z) \right)\nonumber  \\
 &\quad+
  {\left\{ f(t,v -\gamma_{i}(t,v,0),i,1) - f(t,v,i,0)\right\} (1-z){h_{i}}},
  \label{eq:generator}
\end{align}
for each $i=1,\dots,N$. The local martingale component in (\ref{KSMD}) takes the form:
\begin{align}
\mathcal{M}_{t} &:=  \sum_{i=1}^{N}\bigg\{ \int_{0}^{t}\sum_{j\neq{}i} \left[f(u,V_{u^{-}}+{\beta^{0}_{i,j}},j,H(u^{-})) -f(u,V_{u^{-}},i,H(u^{-}))\right] dM_{u}^{i,j} \nonumber \\
    &\quad\quad \quad+  \int_{0}^{t}\left\{f\left(u,V_{u^{-}}-\gamma_{i}(u,V_{u^{-}},0),i, 1\right)-f(u,V_{u^{-}},i,0)\right\}
    {\bf 1}_{\{C_{u^{-}}=i\}}
	d\xi^{\Px}_{u} \nonumber \\
    &\quad\quad\quad +  \int_{0}^{t}
    \vartheta_{i}\frac{\partial f}{\partial v} (u,V_{u},i,{H(u)})H^{i}_{u}dW_{u}\bigg\},
\label{eq:lmt}
\end{align}
where the functions $\beta_{i,j}$, {$\beta_{i,j}^{0}$}, and $\vartheta_{i}$ are evaluated at $(u,V_{u^{-}},H(u^{-}))$ and we used the notation (\ref{JmpTrnPrc1}).

\section{Proof of the verification theorems}\label{SectVerification}

\proof[Proof of Theorem \ref{MntPstDftVal}]
We first note that in the post-default case, the process (\ref{eq:wealtheqsimpl2}) takes the form
\begin{align}\label{GenDyncPstDft}
 dV_s^{\pi,t,v}&= V_s^{\pi,t,v}\bigg\{ \left[ {r_{_{C_{s}}}}+\pi_s ({\mu_{_{C_{s}}}}-{r_{_{C_{s}}}}) \right]  ds+
 \pi_{s} \sigma_{_{C_{s}}}d W_{s}\bigg\},\\
V_{t}^{\pi,t,v} &=  v, \quad \quad(t<s<R).\nonumber
\end{align}
Define the process
\begin{equation}\label{AuxPrc1}
	M_{s}^{\pi} :=  {\underline{w}}(s,V_s^{\pi,t,v},C_{s}), \quad (t\leq{}s\leq{}R),
\end{equation}
for an admissible feedback control $\pi_{s}:={\pi_{_{C_{s}}}}(s,V_{s}^{\pi,t,v})\in\mathcal{A}_{t}(v,i,1)$. For simplicity, through this part we sometimes write $V_{u}$ or $V^{\pi}_{u}$ instead of $V^{\pi,t,v}_{u}$. We prove the result through the following steps:

\medskip
\item[{\bf (1)}]  By the semimartingale decomposition (\ref{KSMD}), it follows that
\[
	 M_{s}^{\pi} = M_{t}^{\pi} + \int_t^s R(u,V_u^{\pi}, C_u, \pi_{u}) du + \mathcal{M}_{s}-\mathcal{M}_{t},
\]
where
\begin{eqnarray}
\nonumber \mathcal{M}_{s} &=& \sum_{i=1}^{N}\bigg\{\sum_{j\neq{}i}\int_{0}^{s} \left\{{\underline{w}}(u,V_{u^{-}},j)
- {\underline{w}}(u,V_{u^{-}},i)\right\} d M_{u}^{i,j}\\
 & & + \int_{0}^{s} \sigma_{i} V_{u} {\pi}_{i}(u,V_{u}) {\underline{w}}_{v}(u,V_{u},i) H_{u}^{i}d W_{u}\bigg\},
 \label{eq:PostDfltMrt}
 \\
\nonumber R(u,v,i,\pi)  &=& {\underline{w}}_u(u, v,i) + {\underline{w}}_v(u, v,i) v  \left( r_i + \pi (\mu_i - r_i) \right) \\
& &  + \frac{1}{2} {\underline{w}}_{vv}(u,v,i) v^2  \pi^2  \sigma_i^2 \nonumber \\
& & + \sum_{j \neq i} a_{i,j}(u) \left({\underline{w}}(u,v,j) - {\underline{w}}(u,v,i) \right).
\label{eq:Requation}
\end{eqnarray}
We have that $R(u,v,i,\pi)$ is a concave function in $\pi$ since, {by assumption},  ${\underline{w}}_{vv} < 0$.
If we maximize $R(u,v,i,\pi)$ as a function of $\pi$ for each $(u,v,i)$, we find that the optimum is given by (\ref{eq:optpi}).
This implies that
\begin{align*}
\nonumber R(u,v,i,\pi)\leq R(u,v,i,{\widetilde\pi_{i}(u,v)}) &= {\underline{w}}_u(u,v,i) + r_i v {w}_v(u,v,i) - \eta_i \frac{{\underline{w}}^2_v(u,v,i)}{{\underline{w}}_{vv}(u,v,i)} \\
& +\sum_{j \neq i} {a_{i,j}} \left({\underline{w}}(u,v,j) - {\underline{w}}(u,v,i) \right)  = 0,
\end{align*}
where the last equality follows from Eq.~(\ref{eq:dirich}). Next, let us introduce the stopping times
\(
	{\tau_{a,b}:=\inf\{s\geq{}t: V_{s}\geq{}{b^{-1}}, \text{ or } V_{s}\leq{}a\}},
\)
for fixed {$0<a<v<b^{-1}<\infty$}. Then, using the notation $\Ex_{t}[\cdot]=\Ex[\cdot|\mathcal{G}_{t}]$, we get the inequality
\begin{align*}
	\Ex_{t}\left[M^{\pi}_{s\wedge \tau_{a,b}}\right]&\leq M_{t}^{\pi} +
	\sum_{i=1,j\neq{}i}^{N}\Ex_{t}\left[\int_{t}^{s\wedge\tau_{a,b}}
	\left\{{\underline{w}}(u,V_{u^{-}},j)- {\underline{w}}(u,V_{u^{-}},i)\right\} d M_{u}^{i,j}\right]\\
	&\quad+\sum_{i=1}^{N}\Ex_{t}\left[\int_{t}^{s\wedge\tau_{a,b}}
	V_{u}\pi_{i}(u,V_{u})\sigma_{i} {\underline{w}}_{v}(u,V_{u},i)H_{u}^{i} d W_{u}\right],
\end{align*}
with equality if $\pi=\widetilde{\pi}$.
Since
\[
	\sup_{t\leq u\leq \tau_{a,b}\wedge R} |{\underline{w}}(u,V_{u},i)|\leq B_{1}, \quad
	\sup_{t\leq u\leq \tau_{a,b}\wedge R} |V_{u}\pi_{i}(u,V_{u}){\underline{w}}_{v}(u,V_{u},i)|^{2}\leq B_{2},
\]
for {some} constants $B_{1},B_{2}<\infty$, we conclude that
\(
	 {\Ex_{t}\left[M^{\pi}_{R\wedge \tau_{a,b}}\right]\leq M_{t}^{\pi}={\underline{w}}(t,v,C_{t})},
\)
with equality if $\pi=\widetilde{\pi}$.

\medskip
\item[{\bf (2)}]
In this step, we show that
\begin{equation}\label{FIdN}
	\lim_{{a,b\to{}0}} \Ex_{t}\left[{\underline{w}}(R\wedge \tau_{a,b},V^{\widetilde{\pi}}_{R\wedge \tau_{a,b}},C_{R\wedge \tau_{a,b}})\right]=
	\Ex_{t}\left[U(V^{\widetilde{\pi}}_{R})\right],
\end{equation}
where $\widetilde{\pi}_{s}=\widetilde{\pi}(s,V_{s}^{\pi,t,v},C_{s})$.
Note that (\ref{KCUB}-i) implies
\begin{align*}
	\Ex\left[\left. \left|{\underline{w}}(R\wedge \tau_{a,b},V^{\widetilde{\pi}}_{R\wedge \tau_{a,b}},C_{R\wedge \tau_{a,b}})\right|^{2}\right|\mathcal{G}_{t}\right]
	&\leq B_{1}+B_{2}\Ex\left[\left. \left|V^{\widetilde{\pi}}_{R\wedge \tau_{a,b}}\right|^{2}\right|\mathcal{G}_{t}\right],
\end{align*}
for some constants $B_{1},B_{2}<\infty$. Next, we note that $\tilde{\pi}$ satisfies (\ref{NCFUB}) since
\[
	\left|v{\tilde\pi_{i}(s,v)}\right|=\left|\frac{\mu_{i}-r_{i}}{\sigma_{i}^{2}} \frac{{\underline{w}}_v(s,v,i)}{{\underline{w}}_{vv}(s,v,i)}\right|<G(s)(1+v),
\]
in light of (\ref{KCUB}-ii). Hence, we can {apply} Lemma \ref{ExpBnd} below (with $\pi^{P}\equiv0$) and obtain
\begin{align*}
	\sup_{0<a<v<b^{-1}<\infty}\Ex_{t}\left[ \left(V^{\widetilde{\pi}}_{R\wedge \tau_{a,b}}\right)^{2}\right]\leq 2\left(V_{t}^{\widetilde{\pi}}\right)^{2}+
	2\Ex_{t}\left[\sup_{t\leq s\leq{}R} \left(V^{\widetilde{\pi}}_{s}-V^{\widetilde{\pi}}_{t}\right)^{2}\right]<\infty.
\end{align*}
Using Corollary 7.1.5 in \cite{ChTei}, we {conclude} (\ref{FIdN}).

\medskip
\item[{{\bf (3)}}]
Finally, if ${\underline{w}}$ is non-negative, then Fatou's Lemma implies that
\begin{align*}
	\Ex_{t}\left[U(V^{\pi}_{R})\right]&= \Ex_{t}\left[
	 \liminf_{{a,b}\to{}0} {\underline{w}}(R\wedge \tau_{a,b},V^{\pi}_{R\wedge \tau_{a,b}},C_{R\wedge \tau_{a,b}})\right]\\
	 & \leq \liminf_{{a,b\to{}0}} \Ex_{t}\left[{\underline{w}}(R\wedge \tau_{a,b},V^{\pi}_{R\wedge \tau_{a,b}},C_{R\wedge \tau_{a,b}})\right]\\
	 &\leq{} {\underline{w}}(t,v,C_{t})=\Ex_{t}\left[U(V^{\widetilde\pi}_{R})\right],
\end{align*}
for every admissible feedback control $\pi_{s}=\pi_{_{C_{s}}}(s,V_{s}^{\pi,t,v})\in\mathcal{A}_{t}(v,i,1)$. For a general function ${\underline{w}}$ (not necessarily non-negative), we proceed along the lines of step (2) above to show
\[
	\lim_{{a,b\to{}0}} \Ex_{t}\left[{\underline{w}}(R\wedge \tau_{a,b},V^{{\pi}}_{R\wedge \tau_{a,b}},C_{R\wedge \tau_{a,b}})\right]=
	\Ex_{t}\left[U(V^{{\pi}}_{R})\right],
\]
for any {feedback control} $\pi_{s}=\pi_{_{C_{s}}}(s,V_{s}^{\pi,t,v})\in\mathcal{A}_{t}(v,i,1)$ satisfying (\ref{NCFUB}).
\endproof

\noindent\proof[Proof of Theorem \ref{MntPstDftVal2}]
We prove the result through the following steps:
\item[{\bf (1)}]
Define the process
\begin{equation}\label{AuxPrc1Pre}
	M_s^{\pi} :=   \bar{w}(s,V_s^{\pi,t,v},C_{s}) {(1-H(s))}+ {\underline{w}(s,V_s^{\pi,t,v},C_{s}) H(s)},
\end{equation}
where $V_s^{\pi,t,v}$ is the solution of Eq.~(\ref{eq:wealtheqsimpl2})
for an admissible feedback control
\[
	 \pi_{s}:=(\pi_{s}^{S},\pi_{s}^{P}):=(\pi^{S}_{_{C_{s^{-}}}}(s,V_{s^{-}}^{\pi,t,v}),\pi^{P}_{_{C_{s^{-}}}}(s,V_{s^{-}}^{\pi,t,v}))\in \mathcal{A}_{t}(v,i,0).
\]
For simplicity, we only write $V_s^{\pi}=V_s^{\pi,t,v}$.
Using the same arguments as in Eq.~(\ref{DcmpCngNotn}) and the decomposition (\ref{MrtRprDftPrcP}), the process (\ref{eq:wealtheqsimpl2}) can be written as
\begin{align}\nonumber
 dV_s^{\pi,t,v}&= V_{{s-}}^{\pi,t,v}\bigg\{ \left[ {r_{_{C_{s}}}}+\pi_{_{C_{{s-}}}}^{S} ({\mu_{_{C_{s}}}}-{r_{_{C_{s}}}}) +\pi^{P}_{_{C_{{s-}}}}(1-H(s))\theta_{C_{s}}(s)\right]  ds\\
 \label{GenDyncPstDftb}
 &\quad +\pi_{_{C_{{s-}}}}^{S} \sigma_{_{C_{s}}}d W_{s}+\pi_{_{C_{s^{-}}}}^{P}(1-H(s^{-})) dH(s)\\
 \nonumber
 &\quad +(1-H(s^{-}))\sum_{i=1}^{N}\sum_{j\neq{}i} \pi_{i}^{P}\frac{{\psi_{j}(s)-\psi_{i}(s)}}{\psi_{i}(s)}d H_{s}^{i,j}\bigg\},
\end{align}
for $s\in(t,R)$, with the initial condition $V_{t}^{\pi,t,v} =  v$, {where $\theta_{i}$ is defined in (\ref{DfnTheta})}.
By the semimartingale decomposition (\ref{KSMD}) with the coefficients given by (\ref{CffDynWlth}) and
\[
	f(s,v,i,z):= \bar{w}(s,v,i) (1-z)+ \underline{w}(s,v,i) z,
\]
it follows that
\[
	 M_{s}^{\pi} = M_{t}^{\pi} + \int_t^s R(u,V_u^{\pi}, C_u,\pi^{S}_{_{C_{u}}}, \pi_{_{C_{u}}}^{P}, {H(u)}) du + \mathcal{M}_{s}-\mathcal{M}_{t}.
\]
Here, $\mathcal{M}_{s}-\mathcal{M}_{t}$ is given by
\begin{align}
\nonumber &\sum_{i=1}^{N}\bigg\{\sum_{j\neq{}i}\int_{t}^{s} \left[{w}\left(u,V_{u^{-}}\left(1+\bar{\pi}^{P}_{i}\frac{{\psi_{j}(u)-\psi_{i}(u)}}{\psi_{i}(u)}\right),j\right)- {w}(u,V_{u^{-}},i)\right] d M_{u}^{i,j}\\
\nonumber &\quad \quad + \int_{t}^{s} \left[{w}\left(u,V_{u^{-}}\left(1-\bar{\pi}^{P}_{i}\right),i\right)- {w}(u,V_{u^{-}},i)\right] {\bf 1}_{\{C_{u^{-}}=i\}}d \xi^{\Px}_{u}\\
& \quad \quad + \int_{t}^{s} \sigma_{i} V_{u} {\pi}^{S}_{i}(u,V_{u}) {w}_{v}(u,V_{u},i) H_{u}^{i}d W_{u}\bigg\},
 \label{eq:PostDfltMrt}
 \end{align}
where  $\bar\pi^{P}_{i}(u,v):=(1-H(u))\pi^{P}_{i}(u,v)$ and
\[
	{w(u,v,i)=\bar{w}(u,v,i)(1-H(u^{-}))+\underline{w}(u,v,i)H(u^{-})}.
\]
Similarly, $R(u,v,i,\pi^{S},\pi^{P},1)$ is defined as in (\ref{eq:Requation}) with $\pi=\pi^{S}$, while
 \begin{align}
\nonumber  R(u,v,i,\pi^{S},\pi^{P},0)  &:= \bar{w}_u(u, v,i) + \bar{w}_v(u, v,i) v  \bigg( r_i + \pi^S (\mu_i - r_i)
+ \pi^P \theta_{i}(u)\bigg) \\
\nonumber &   + \frac{1}{2} \bar{w}_{vv}(u,v,i) v^2 (\pi^S)^2 \sigma_i^2 \\
\nonumber & + \sum_{j \neq i} a_{i,j}(u)\left[\bar{w}\left(u,v \left( 1 + \pi^P \frac{{\psi_j(u)-\psi_i(u)}}{\psi_i(u)} \right),j \right)
- \bar{w}(u,v,i)  \right] \\
 &  +  h_i \left[ {\underline{w}}\left(u, v \left( 1 - \pi^P \right),i \right) - \bar{w}\left(u, v,i \right) \right].
\label{eq:Requationpre}
\end{align}
Note that, under our assumptions, $(\pi^{S},\pi^{P})\to R(u,v,i,\pi^{S},\pi^{P},1)$ admits a unique maximal point $(\widetilde{\pi}^{S},\widetilde{\pi}^{P})$ for each $(u,v,i)$, since $(i)\; R_{\pi^S\pi^S} \leq 0$,
 $(ii)\; R_{\pi^P\pi^P} \leq 0$, and  $(iii)\;R_{\pi^S\pi^P} =0$. Indeed, $(i)$ follows from our assumption that $\bar{w}_{vv} \leq  0$, while $(ii)$ is evident from the calculation
\begin{align*}
\nonumber R_{\pi^P \pi^P} & = \sum_{j \neq i} a_{i,j}(u)\frac{({\psi_j(u)-\psi_{i}(u)})^2}{\psi_i(u)^2} v^2 \bar{w}_{vv}\left(u, v \left( 1 + \pi^P \frac{{\psi_j(u)-\psi_i(u)}}{\psi_i(u)} \right),j \right)\\
&\quad + h_{i} v^2 {\underline{w}}_{vv} \left(u, v \left(1-\pi^P \right),i\right),
\end{align*}
and the fact that ${\underline{w}_{vv}}\leq 0$.
The optimum is given by Eq.~(\ref{OptPrDf}) with $p$ defined implicitly by Eq.~(\ref{eq:optpiP}). In light of Eqs.~(\ref{eq:dirich})-(\ref{eq:dirichpre}),
\[
	R(u,v,i,\pi^{S},\pi^{P},z)\leq  R(u,v,i,\widetilde{\pi}^{S}_{i}(u,v),\widetilde{\pi}^{P}_{i}(u,v),z)=0.
\]
Let
\begin{align*}
	\tau_{a,b}:=\inf\bigg\{&s\geq{}t: V_{s}\geq{}{b^{-1}}, \text{ or } V_{s}\leq{}a,   \text{ or } {{\pi}^{P}_{u^{-}}>1-a}, \\
	&\quad  \text{ or }\; {\pi}^{P}_{s^{-}}{\frac{\psi_{j}(s)-\psi_{i}(s)}{\psi_{i}(s)}<a-1,\text{ for some }j\neq{}i}\bigg\},
\end{align*}
for a small enough $a,b>0$.
Using similar arguments to those in the proof of Theorem \ref{MntPstDftVal},
we can show that
\[
	\Ex_{t}\left[M^{\pi}_{R\wedge \tau_{a,b}}\right]\leq M_{t}^{\pi}=\bar{w}(t,v,C_{t}) (1-H(t))+ \underline{w}(t,v,C_{t}) H(t),
\]
with equality if $\pi^S=\widetilde{\pi}^{S}$ and $\pi^P=\widetilde{\pi}^{P}$.

\item[{\bf (2)}]
We now show that
\begin{align}
	\label{AL1VP}
	\lim_{{a,b}\to{}0} \Ex_{t}\left[\bar{w}(R\wedge \tau_{a,b},V^{\widetilde{\pi}}_{R\wedge \tau_{a,b}},C_{R\wedge \tau_{a,b}})\bar{H}(R\wedge \tau_{a,b})\right]&=
	\Ex_{t}\left[U(V^{\widetilde{\pi}}_{R})\bar{H}(R)\right]\\
	\lim_{{a,b\to{}0}} \Ex_{t}\left[\underline{w}(R\wedge \tau_{a,b},V^{\widetilde{\pi}}_{R\wedge \tau_{a,b}},C_{R\wedge \tau_{a,b}}){H}(R\wedge \tau_{a,b})\right]&=
	\Ex_{t}\left[U(V^{\widetilde{\pi}}_{R}){H}(R)\right], \label{AL2VP}
\end{align}
where $\bar{H}(s)=1-H(s)$. Note that (\ref{AL1VP}-\ref{AL2VP}) will imply that
\begin{equation*}
	\bar{w}(t,v,C_{t}) \bar{H}(t)+ \underline{w}(t,v,C_{t}) H(t)=\lim_{{a,b\to{}0}} \Ex_{t}\left[M^{\widetilde{\pi}}_{R\wedge \tau_{a,b}}\right]=\Ex_{t}\left[U(V^{\widetilde{\pi}}_{R})\right].
\end{equation*}
We only {prove (\ref{AL1VP}) ((\ref{AL2VP}) can be treated similarly)}. Note that (\ref{KCUB}-i) implies
\begin{align*}
	\Ex\left[\left. \left|\bar{w}(R\wedge \tau_{a,b},V^{\widetilde{\pi}}_{R\wedge \tau_{a,b}},C_{R\wedge \tau_{a,b}})\bar{H}(R\wedge \tau_{a,b})\right|^{2}\right|\mathcal{G}_{t}\right]
	&\leq B_{1}+B_{2}\Ex\left[\left. \left|V^{\widetilde{\pi}}_{R\wedge \tau_{a,b}}\right|^{2}\right|\mathcal{G}_{t}\right],
\end{align*}
for some constants $B_{1},B_{2}<\infty$. Next, we note that $\tilde{\pi}^{S}$ satisfies (\ref{NCFUB}) since both $\bar{w}$ and $\underline{w}$ satisfies (\ref{KCUB}-ii) by assumption. Also, $\widetilde{\pi}^{P}$ satisfies (\ref{NcsCndAdm}) by assumption. Hence, we can applying Lemma \ref{ExpBnd} below and obtain
\begin{align*}
	\sup_{0<a<v<b^{-1}<\infty}\Ex_{t}\left[ \left(V^{\widetilde{\pi}}_{R\wedge \tau_{a,b}}\right)^{2}\right]\leq 2\left(V_{t}^{\widetilde{\pi}}\right)^{2}+
	2\Ex_{t}\left[\sup_{t\leq s\leq{}R} \left(V^{\widetilde{\pi}}_{s}-V^{\widetilde{\pi}}_{t}\right)^{2}\right]<\infty.
\end{align*}
Using Corollary 7.1.5 in \cite{ChTei}, we {conclude (\ref{AL1VP})}.

\item[{\bf (3)}]
Finally, if $w$ is non-negative, then Fatou's Lemma implies that
\begin{align*}
	\Ex_{t}\left[U(V^{\pi}_{R})\right]&= \Ex_{t}\left[
	 \liminf_{{a,b\to{}0}} \bar{w}(R\wedge \tau_{a,b},V^{\pi}_{R\wedge \tau_{a,b}},C_{R\wedge \tau_{a,b}})\bar{H}(R\wedge \tau_{a,b})\right]\\
	 &\quad +\Ex_{t}\left[
	 \liminf_{{a,b\to{}0}} \underline{w}(R\wedge \tau_{a,b},V^{\pi}_{R\wedge \tau_{a,b}},C_{R\wedge \tau_{a,b}}){H}(R\wedge \tau_{a,b})\right]\\
	 & \leq \liminf_{{a,b\to{}0}} \Ex_{t}\bigg[\bar{w}(R\wedge \tau_{a,b},V^{\pi}_{R\wedge \tau_{a,b}},C_{R\wedge \tau_{a,b}})\bar{H}(R\wedge \tau_{a,b})\\
	 &\quad\quad\quad\quad \quad\quad +
	 \underline{w}(R\wedge \tau_{a,b},V^{\pi}_{R\wedge \tau_{a,b}},C_{R\wedge \tau_{a,b}}){H}(R\wedge \tau_{a,b})\bigg]\\
	 &\leq{} \bar{w}(t,v,C_{t})\bar{H}(t)+\underline{w}(t,v,C_{t}){H}(t)=\Ex_{t}\left[U(V^{\widetilde\pi}_{R})\right],
\end{align*}
for every admissible feedback control {$\pi_{s}=\pi_{_{C_{s^{-}}}}(s,V_{s^{-}}^{\pi,t,v},H(s^{-}))\in\mathcal{A}_{t}(v,i,0)$}. For a general function $w$ (not necessarily non-negative), we proceed along the lines of step (2) above to show
\begin{align*}
	\lim_{{a,b\to{}0}} &\Ex_{t}\bigg[\bar{w}(R\wedge \tau_{a,b},V^{{\pi}}_{R\wedge \tau_{a,b}},C_{R\wedge \tau_{a,b}})\bar{H}(R\wedge \tau_{a,b})\\
	&\quad+\underline{w}(R\wedge \tau_{a,b},V^{{\pi}}_{R\wedge \tau_{a,b}},C_{R\wedge \tau_{a,b}}){H}(R\wedge \tau_{a,b})\bigg]=
	\Ex_{t}\left[U(V^{{\pi}}_{R})\right],
\end{align*}
for any $t$-admissible feedback controls  {$\pi_{s}=\pi_{_{C_{s^{-}}}}(s,V_{s^{-}}^{\pi,t,v},H(s^{-}))\in\mathcal{A}_{t}(v,i,0)$} satisfying (\ref{NCFUB}) and (\ref{NcsCndAdm}).
\endproof

\begin{lemma}\label{ExpBnd}
	Let {$\pi^{S}_{i}(s,v,z)$ and $\pi^{P}_{i}(s,v,z)$} be functions such that (\ref{GenDyncPstDftb})  admits a unique nonnegative solution {$(V_{s}^{\pi})_{s\in[t,R]}$}. We also assume that $\pi^{S}$ satisfies (\ref{NCFUB}) and $\pi^{P}$ satisfies (\ref{NcsCndAdm}).  Then,  the solution of (\ref{GenDyncPstDftb}) satisfies the moment condition:
	\begin{equation}\label{UBCPD}
		\Ex_{t}\left[\sup_{t\leq{}u\leq{}R}|{V^{\pi}_{u}-V^{\pi}_{t}}|^{2}\right]\leq C_{1} {(V_{t}^{\pi})^{2}}+ C_{2},
	\end{equation}
	for some constants $C_{1},C_{2}<\infty$.
\end{lemma}
\begin{proof}
{For simplicity, we write $V_{s}$ instead of $V^{\pi}_{s}$}.
Let us start by
{recalling} that we can write (\ref{GenDyncPstDftb}) in the form
\begin{align}\label{GenDynb}
 dV_s= \alpha_{_{C_{s}}} ds+ \vartheta_{_{C_{s}}} d W_{s}
 +\sum_{j=1}^{N}\beta_{_{C_{s^{-}},j}} dM^{\Px}_{j}(s)
- \gamma_{_{C_{s^{-}}}} d\xi^{\Px}_{s},
\end{align}
taking the coefficients as in (\ref{CffDynWlth}). Due to {(\ref{NcsCndAdm}) and (\ref{NCFUB})}, we can see that $|\alpha_{i}(s,v,z)|\leq{}E(s)(1+v)$ for a locally bounded function $E$. Hence, by Jensen's inequality and the previous linear growth,
\begin{align*}
	\left|\int_{t}^{s}\alpha_{_{C_{u}}}(u,V_{u},H(u^{-})) du \right|^{2}
	&\leq \kappa \tau \left(\tau +\tau V_{t}^{2}+ \int_{t}^{s} |V_{u}-V_{t}|^{2} du\right),
\end{align*}
{for any $s\in[t,R]$, where} $\tau:=R-t$ and $\kappa$ denote a generic constant that may change from line to {line}.
Similarly, denoting $\tau_{b}=\inf\{s\geq{}t: |V_{s}|\geq{}b\}$ ($b>v$) and using Burkh\"older-Davis-Gundy inequality (see Theorem 3.28 in {\cite{karatSh} or Theorem IV.48 in \cite{Protter}}) and Jensen's inequality,
\begin{align*}
	\Ex_{t} \sup_{t\leq{}s\leq{}R\wedge \tau_{b}} \left|\int_{t}^{s}  \vartheta_{_{C_{u}}}(u,V_{u},H(u))
	d W_{u}\right|^{2}&\leq \kappa\tau \Ex\int_{t}^{R\wedge \tau_{b}} \left|  \sigma_{_{C_{u}}} \pi^{S}_{u} V_{u}
	\right|^{2}du.
\end{align*}
We can then again use (\ref{NCFUB}) to show that
\begin{align*}
	\Ex_{t} \sup_{t\leq{}s\leq{}R\wedge \tau_{b}} \left|\int_{t}^{s} \vartheta_{_{C_{u}}}(u,V_{u},H(u))d W_{u}\right|^{2}&\leq{} \kappa\tau (\tau +\tau |V_{t}|^{2}+ \int_{t}^{R\wedge \tau_{b}} |V_{u}-V_{t}|^{2} du).
\end{align*}
Next, using again Burkh\"older-Davis-Gundy inequality {(see, e.g., Theorem 23.12 in \cite{Kallenberg})},
\begin{align*}
	\Ex_{t} \sup_{t\leq{}s\leq{}R\wedge \tau_{b}}
	\left| \int_{t}^{s}\gamma_{_{C_{u^{-}}}}(u,V_{u^{-}},H(u^{-})) d\xi^{\Px}_{u}\right|^{2}&\leq
	\kappa \Ex_{t} \int_{t}^{R\wedge\tau_{b}}| V_{u^{-}}\pi_{u}^{P}|^{2} d H(u).
\end{align*}
Using  (\ref{NcsCndAdm}),
\begin{align*}
	\Ex_{t} \int_{t}^{R\wedge\tau_{b}}| V_{u^{-}}\pi_{u}^{P}|^{2} d H(u)&\leq{}
	\kappa \Ex_{t} \int_{t}^{R\wedge\tau_{b}}| V_{u^{-}}|^{2} d H(u)\\
	&=
	\kappa \Ex_{t} \int_{t}^{R\wedge\tau_{b}}| V_{u}|^{2} h_{_{C_{u}}}(1-H(u))d u.
\end{align*}
Then, we can proceed as before to conclude that
\begin{align*}
	\Ex_{t} \sup_{t\leq{}s\leq{}R\wedge \tau_{b}}
	\left| \int_{t}^{s}\gamma_{_{C_{u^{-}}}}(u,V_{u^{-}},H(u^{-})) d\xi^{\Px}_{u}\right|^{2}
	\leq \kappa\tau (\tau +\tau |V_{t}|^{2}+ \int_{t}^{R\wedge \tau_{b}} |V_{u}-V_{t}|^{2} du).
\end{align*}
Using a similar argument, we can also obtain that
\begin{align*}
	\Ex_{t} \sup_{t\leq{}s\leq{}R\wedge \tau_{b}}
	\left| \int_{t}^{s}\beta_{_{C_{u^{-}},j}}(u,V_{u^{-}},H(u^{-})) d M_{j}^{\Px}(u)\right|^{2}
	\leq \kappa\tau (\tau +\tau |V_{t}|^{2}+ \int_{t}^{R\wedge \tau_{b}} |V_{u}-V_{t}|^{2} du).
\end{align*}
Putting together the previous estimates, we conclude that the function
$\gamma_{b}(r):=\Ex_{t} \sup_{t\leq{}s\leq{}r\wedge \tau_{b}}|V_{s}-V_{t}|^{2}$ can be bounded as follows:
\(
	\gamma_{b}(r)\leq \kappa\tau(1+v^{2}) +\kappa \int_{t}^{r} \gamma_{b}(u) du.
\)
By Gronwall inequality, we have
\[
	\gamma_{b}(R)\leq{} \kappa (R-t)(1+ v^{2})e^{\kappa(R-t)},
\]
and (\ref{UBCPD}) is obtained by making $b\to\infty$.
\end{proof}

\section{Proof of Explicit Constructions}\label{SectExplConstr}

\proof[Proof of Lemma \ref{lemma:pcontin}]
For fixed $i$ and $s$, consider the function
$$
	f({p_{i}},s,i):= \theta_i(s) - \frac{h_i}{1-{p_{i}}} + \sum_{j \neq i} a_{i,j}(s) \frac{{\psi_j(s)-\psi_i(s)}}{\psi_i(s) + {p_{i}} ({\psi_j(s)-\psi_{i}(s)})}.$$
We first observe that $f(p_{i},i,s)$ is a continuous function of $p_{i}$ in the interval $(M_i,1)$.
Indeed, we can write the above summation as
\[
	{\sum_{j \neq i: \psi_{j}(s)>\psi_{i}(s)}  \frac{a_{i,j}(s)}{\frac{\psi_i(s)}{{\psi_j(s)-\psi_{i}(s)}} + {p_{i} }}+\sum_{j \neq i: \psi_{j}(s)< \psi_{i}(s)}  \frac{a_{i,j}(s)}{\frac{\psi_i(s)}{{\psi_j(s)-\psi_{i}(s)}} + {p_{i} }}},
\]	
and since $1<\frac{-\psi_i(s)}{{\psi_j(s)-\psi_{i}(s)}}$ when $\psi_{j}(s)<\psi_{i}(s)$, we have
$p_{i}+\frac{\psi_i(s)}{{\psi_j(s)-\psi_{i}(s)}}<0$ for $p_{i}\in(M_{i},1)$.
Moreover, the previous decomposition also shows for each fixed $s$, $f(p_{i},s,i)$ is strictly decreasing in $p_{i}$ from $(M_{i},1)$ onto $(-\infty,\infty)$. This implies the existence of a unique $p_{i}(s)$ such that $f(p_i(s),s,i) = 0$, for any $s>0$. In light of \cite{Kumagai} implicit theorem, we will also have that $p_{i}(s)$ is continuous if we prove that $f(p_{i},s,i)$ is continuous in $(p_{i},s)$. The latter property follows because, by assumption, $a_{i,j}$ and $a_{i,j}^{\Qx}$ are continuous, which implies directly the continuity of the functions $\theta_i$. The continuity of the functions $\psi_j$ will follow from a similar argument to that of Lemma \ref{Lm:DfPsi}.
\endproof

\proof[Proof of Proposition \ref{prop:logvaluefn}] 

\medskip

\item[{\bf (1)}] It can be checked that the function $\underline{\varphi}_t^R(t,v,i) = \log(v) + K(t,i)$ solves the Dirichlet problem
~(\ref{eq:dirich}), if the functions $K(t,i)$, $i = 1,\ldots,N$ satisfy the system given by Eq.~(\ref{eq:postdeflog}), which may be written in matrix-vector form as
\begin{align}
\label{eq:ktilog}
 K_t(t) &= F(t) K(t) + b(t), \qquad K(R) = 0,
\end{align}
where  $[F(t)]_{i,j} = -a_{i,j}(t)$ and $b(t) = -\frac{\eta_i^2}{2} - r_i$.
As $a_{i,j}$ is continuous in $[0,T]$ by hypothesis, we have that the system admits a unique solution in $[0,R]$ by Lemma \ref{lemmaode}, where $R \leq T$. Moreover, $\underline{\varphi}_t^R(t,v,i) \in C_{1,2}^0$ due to concavity and increasingness of the logarithmic function,
and, under the choice ${D(t)} = \max_{i=1,\dots,N}K(t,i)$ and $G(t) = 1$, the function $\underline{\varphi}_t^R(t,v,i)$ satisfies the conditions in ~(\ref{KCUB}). Therefore, applying Theorem \ref{MntPstDftVal}, we can conclude that, for each $i = 1,\ldots,N$, $\underline{\varphi}_t^R(t,v,i)$ is the optimal post-default value function.

\item[{\bf (2)}] Plugging the expression for $\underline{\varphi}_t^R(t,v,i)$ inside Eq.~(\ref{eq:optpi}), we can conclude immediately that $\tilde{\pi}^S(t,j) = \frac{\mu_j - r_j}{\sigma_j^2}$.
\item[{\bf (3)}] It can be checked that the vector of functions $[\overline{\varphi}_t^R(t,v,1),\ldots,\overline{\varphi}_t^R(t,v,{N})]$, where $\overline{\varphi}_t^R(t,v,i) = \log(v) + J(t,i)$, and the vector $p(t)$ solving the nonlinear system of equations~(\ref{eq:optimalpi}) simultaneously satisfy the system composed of Eq.~(\ref{eq:optpiP}) and Eq.~(\ref{eq:dirichpre}) if the vector of functions $J(t) = (J(t,1),J(t,2),\ldots, J(t,N){)}$  solves the system~(\ref{eq:predeflog1}), which may be written in matrix form as
 \begin{align}
 \nonumber J_t(t) &= F(t) J(t) + b(t), \qquad J(R)=0 \\
 \nonumber [F(t)]_{i,j} &= -a_{i,j}(t) \qquad (j \neq i), \quad [F(t)]_{i,i} = h_i - a_{i,i}(t) \\
 \nonumber b_i(t) &= -\frac{\eta_i^2}{2} - r_i - {p_i(t)} \theta_i(t) - h_i \left(\log(1- {p_i(t)}) + {K(t,i)} \right) \\
\nonumber &\quad \quad - \sum_{j \neq i} a_{i,j}(t) \log\left( 1 + {p_i(t)} \frac{{\psi_j(t)-\psi_i(t)}}{\psi_i(t)}  \right)
 \end{align}
   It remains to show that such solution vector $J(t)$ is unique. As we know that ${p_i(t)}$ is continuous in $t$, then we have that each entry $[F(t)]_{i,j}$ is a continuous function of $t$. Moreover, the vector $K(t)$ is continuous, as it solves the system of differential equations given by (\ref{eq:ktilog}). As $a_{i,j}^{\Qx}$ is continuous by assumption, then $\psi_j$'s are continuous, thus $\theta_i$ is continuous, and consequently $b_i(t)$ is continuous in $[0,{R}]$. Using Lemma \ref{lemmaode}, we obtain that the solution vector $J(t)$ must be unique.
   Moreover, under the choice of {$D(t) = \max_{i=1\dots,N}J(t,i)$} and $G(t) = 1$, we have that $\underline{\varphi}_t^R(t,v,i)$ satisfies the conditions in ~(\ref{KCUB}). As the logarithmic function is increasing and concave in $v$, then $\overline{\varphi}_t^R(t,v,i) \in C_0^{1,2}$, therefore it must be the optimal pre-default value function by Theorem \ref{MntPstDftVal2}.
\endproof

\end{document}